\documentclass[aps,preprint,prl,groupedaddress]{revtex4-2}
\usepackage{amsfonts, amsmath, amsthm, amssymb}
\usepackage{caption}
\usepackage{hyperref}
\usepackage{graphicx}
\usepackage{setspace}

\newtheorem{theorem}{Theorem}[]

\newtheorem{corollary}[theorem]{Corollary}

\theoremstyle{definition}
\numberwithin{equation}{section}

\singlespace

\begin{document}

\title{Graph topology determines coexistence in the rock-paper-scissors system}

\author{Mark Lowell}
\email{MarkLowell@theorem-engine.org}
\affiliation{Theorem Engine, PO Box 22479, 368 S Pickett St., Alexandria, VA 
22304}
\homepage{https://www.theorem-engine.org}

\date{\today}

\begin{abstract}

Species survival in the $(3, 1)$ May-Leonard system is determined by the 
mobility, with a critical mobility threshold between long-term coexistence and 
extinction. We show experimentally that the critical mobility threshold is 
determined by the topology of the graph, with the critical mobility threshold 
of the periodic lattice twice that of a topological 2-sphere. We use
topological techniques to monitor the evolution of patterns in various graphs
and estimate their mean persistence time, and show that the difference in
critical mobility threshold is due to a specific pattern that can form on the
lattice but not on the sphere. Finally, we release the software toolkit we
developed to perform these experiments to the community.

\end{abstract}


\maketitle

\section{Introduction}

Over the last thirty years, theoretical ecologists have explored the role of 
space in biodiversity using stochastic spatial models \cite{usersguide}.
A stochastic spatial model consists of a graph, usually a two-dimensional
square lattice, where each vertex has an associated state selected from a
finite set, with each state representing the vertex being occupied by an
individual of a different species. The vertex states evolve stochastically as a
continuous-time Markov chain through interactions between neighbors.

The $(3, 1)$ May-Leonard system, also called the rock-paper-scissors system, is 
one of the most extensively studied of these stochastic spatial models 
\cite{Reichenbach}\cite{May-Leonard}\cite{nr}. Each vertex in the $(3, 1)$ May-
Leonard system can have one of four possible states, where state 0 represents 
an empty vertex, while states 1, 2, and 3 represent different species. The 
system evolves through diffusion, reproduction, and predation reactions, where 
species 1 predates on species 2 which predates on species 3 which predates on 
species 1. This non-transitive competition is one possible explanation for the
paradox of the plankton \cite{plankton} - the mystery of why there are so many
more species in the world than simple mathematical models suggest there should
be - and has been observed in many real-world ecosystems, including \emph{E.
Coli} bacteria \cite{ecoli2}, Californian lizards \cite{lizards}, and coral
reef invertebrates \cite{coral}. On a fully-connected graph, two of the three
species in the $(3, 1)$ May-Leonard system quickly die out. However, on a
lattice, spiral-shaped domains of a single species form, as shown in
FIG.~\ref{fig:example}. Each pixel in FIG.~\ref{fig:example} is a vertex on a
periodic lattice after the system has run long enough for the pattern to
stabilize, with the color denoting the species occupying that vertex. Species 1
chases species 2 across the lattice, species 2 chases species 3, and species 3
chases species 1, with rotating spirals forming where the three domains meet.
These domains can form a variety of distinct patterns on the graph, with each
picture in FIG.~\ref{fig:example} showing an example of a different pattern.

The mobility of a $(3, 1)$ May-Leonard system is defined in \cite{Reichenbach} 
to be the diffusion rate normalized by the system size. The system has a 
critical mobility threshold: if the mobility is below the threshold, the model 
enters a super-persistent transient state and coexistence may be preserved for 
arbitrarily long time. As the mobility increases, the domains grow wider. Above
 the critical threshold, the domains outgrow the lattice and two species die 
out.

However, \cite{Reichenbach} only studied the critical mobility threshold on the 
periodic square lattice. Other work has examined random graphs, and lattices 
with different local neighborhoods \cite{other-graph-1} \cite{other-graph-2} 
\cite{other-graph-3} \cite{other-graph-4}. Aside from random graphs, we are 
unaware of past work that has examined the role of the \emph{global} graph 
structure on the critical mobility threshold. A non-periodic lattice was 
examined in \cite{nonperiodic}, and in fact showed a slightly different 
critical mobility threshold than \cite{Reichenbach}, but they did not connect
this difference to the graph structure.

In this work, we examine the effect of global graph topology on the critical 
mobility threshold, and show that different graphs have different thresholds. 
Each graph we study is \emph{locally} equivalent, but globally distinct. In 
addition, we use topological data analysis to track the patterns that form on 
each graph, and show that these differences in critical mobility threshold are 
caused by topological constraints on which patterns can form on which graph.

Finally, with a few exceptions, e.g. \cite{code}, previous researchers of 
stochastic spatial models have generally not made the code they used publicly 
available. We describe a new software tool for stochastic spatial models, named
 Viridicle, which we used to perform the experiments in this paper, and which 
we release to the community as open source 
\footnote{\url{https://github.com/TheoremEngine/Viridicle}}. We additionally 
make our experimental scripts 
\footnote{\url{https://github.com/TheoremEngine/graph-topology-determines-
survival}} and experimental data files \footnote{\url{https://www.theorem-
engine.com/papers/001/index.html}} public.

In Methodology, we describe Viridicle and the analytic methods we use. In 
Experiments, we use these techniques to explore pattern formation and survival 
in the $(3, 1)$ May-Leonard system on the periodic lattice, on the 
discretization of the topological 2-sphere, and on the discretization of the 
Klein bottle. We show experimentally that certain patterns can form only on 
certain graphs, that each pattern has a distinct median survival duration, and 
that these result in different critical mobility thresholds. In Discussion, we 
follow up these experiments by providing a proof that a particular long-lived 
pattern, which we call the marching bands, cannot form on the sphere, and that 
this explains its lower critical mobility threshold. In Conclusion, we conclude
 by summing up our results and discussing directions for future research.

\section{Methodology}
\label{sec:methodology}

\subsection{May-Leonard Models}

A stochastic spatial model is a graph where, at time $t$, each vertex $v$ has an
 associated state $s(v, t)$ taken from a finite set $\mathcal{S} = \{0, 1, 2, 
..., n\}$, and the vertex states evolve over time through pair reactions. 
(Since the units of time are arbitrary in this model, we treat time as 
unitless.) We assume without loss of generality that our graph is directed, 
replacing the edges in a non-directed graph with pairs of directed edges, one 
in each direction. Given a directed edge $e = (v_1, v_2)$, where $s(v_1, t) = 
s_1$ and $s(v_2, t) = s_2$, the time until vertices $v_1, v_2$ change state to 
$s'_1, s'_2$, conditioned on no other event occuring first, is exponentially 
distributed with rate $\rho(s_1, s_2, s'_1, s'_2)$, where $\rho$ is a rate 
tensor \footnote{There is some disagreement in the literature on the best 
approach to defining rates. It is common to define the rate as the rate at 
which a transition occurs on a \emph{vertex}, e.g., \cite{Reichenbach}. We 
choose to define the rate as the rate at which an event occurs on a directed 
edge because this allows significant computational savings: we can sample the 
next vertex pair by sampling from a list of directed edges, which requires a 
single RNG call, rather than by first sampling a vertex and then sampling a 
neighbor of that vertex, which requires two. The two approachs are equivalent 
up to a rescaling of time on a regular graph, and all graphs we explore are 
regular.}. $\rho$ is typically very sparse, and the possible reactions are often
 described using notation from chemical reactions:
\[
	(s_1, s_2) \to (s_1', s_2')\mbox{ at rate }\rho(s_1, s_2, s'_1, s'_2)
\]
Stochastic spatial models are implemented using the Gillespie algorithm 
\cite{gillespie}. Given a finite collection of possible events $E_i$, where the
 time to each event is exponentially distributed with rate $\rho_i$, the time 
until \emph{some} event occurs is exponentially distributed with rate $\rho_1 +
 \rho_2 + ...$, and the probability that the next event is $E_i$ is $\rho_i / 
(\rho_1 + ...)$. We implement this in a stochastic spatial model by first 
randomly sampling an edge, then randomly sampling a reaction on that edge. If 
the randomly selected edge's vertices $v_1, v_2$ are currently in states $s_1, 
s_2$, then the probability they transition to states $s_1', s_2'$ is:
\[
	\mathbb{P}((s_1, s_2) \to (s_1', s_2')) = \frac{\rho(s_1, s_2, s_1', 
s_2')}{\rho_{\mbox{max}}}
\]\[
	\rho_{\mbox{max}} = \max_{s_1, s_2} \left(\sum_{s_1', s_2'} \rho(s_1, s_2, 
s_1', s_2')\right)
\]
Each Gillespie algorithm step accounts for $\rho_{\mbox{max}}N^{-1}$ time, where
 $N$ is the number of vertices in the graph.\footnote{Some past work, e.g. 
\cite{Reichenbach}, instead defines one step as taking $N^{-1}$ time; the two 
approaches are equivalent up to rescaling of time.}

We follow \cite{nr} in defining a \textbf{$(3, 1)$ May-Leonard model} as a 
stochastic spatial model with 4 states. State $0$ denotes an empty vertex, 
while $1, 2, 3$ all denote vertices occupied by a single individual of that 
species. The model evolves by three reactions, reproduction, diffusion, and 
predation:
\[
	(s, 0)\to (s, s)\mbox{ at rate }0.2\mbox{ for }s > 0
\]\[
	(s_1, s_2)\to (s_2, s_1)\mbox{ at rate }\mu
\]\[
	(s_1, s_2)\to (s_1, 0)\mbox{ at rate }0.2
	\mbox{ if }s_1 + 1\equiv s_2 (\mbox{mod } 3)\mbox{ and }s_1, s_2 > 0
\]
Where the diffusion rate $\mu$ is a hyperparameter. Following 
\cite{Reichenbach3}, we define the \textbf{mobility} to be $m = \mu(2N)^{-1}$, 
where $N$ is the number of vertices. The mobility is derived from a scaling 
term in the asymptotic equivalence of the stochastic spatial model to a PDE,
under the assumption our graph consists locally of 2-dimensional von Neumann 
neighborhoods. Due to differences in how we define the rates and time, our 
mobilities are not directly comparable to the mobilities in 
\cite{Reichenbach3}, but our results are equivalent up to rescaling.

\subsection{Viridicle}

To perform our experiments, we have written and make publicly available a 
software library for simulation of stochastic spatial models, Viridicle. 
Viridicle is a Python libary written in C, using the NumPy library \cite{numpy}
 to allow the user to conveniently pass data to and from the underlying C 
layer.

Viridicle provides a fast implementation of arbitrary stochastic spatial models.
 It uses a lookup table of pointers to sites to represent the graph edges. This
 allows the code to support arbitrary, non-lattice graphs without changes to 
the core functionality. The user can specify a graph structure using the 
NetworkX libary \cite{networkx}. We include in the software example 
implementations of the $(N, r)$ May-Leonard model \cite{Reichenbach} \cite{nr},
 the Durrett \& Levin model of \emph{E. Coli} bacteria \cite{ecoli}, and the 
Avelino et al model of $Z_N$ Lotka-Volterra competition \cite{strings}. We show
 example results for each in FIG.~\ref{fig:examples}; each picture shows a 
sample graph from one of these models after the model has run long enough for 
patterns to stabilize. Each pixel represents a vertex in the graph, and 
different colors represent different model states.

\subsection{Pattern Identification}
\label{subsec:TDA}

As shown in FIG.~\ref{fig:example}, the May-Leonard system can form multiple 
distinct patterns on the lattice. To track the evolution of the pattern over 
the course of the experiment, we treat the graph as a planar graph and 
calculate the Betti numbers of the planar subgraphs consisting of all vertices 
of a specific state. The Betti numbers are invariants of a topological space,
equal to the rank of the simplicial homology groups; see \cite{Hatcher} for an
introduction.

We first clean the graphs to eliminate small imperfections, such as solitary 
empty vertices carried into the interior of a domain by diffusion. We define a 
\textbf{cluster} to be a maximal connected subgraph whose vertices are all the 
same state. We identify all clusters falling below $N / 256$ vertices, where 
$N$ is the number of vertices in the graph. We set the vertices of those 
clusters to a null state, distinct from the states 0 through 3 in the model. We
 then take every vertex in the null state that is adjacent to a vertex of 
another state, and which is adjacent only to vertices of that state or of null 
state, and set that vertex to the neighboring vertex's state. We repeat this 
process $\sqrt{N} / 8$ times. An example of this procedure is shown in 
FIG.~\ref{fig:cleaning}. 

We then calculate the 0th and 1st Betti numbers of the subgraphs of all
vertices belonging to each state. In a closed 2-dimensional manifold, $b_0$ is
equal to the number of clusters, and $b_1$ can be calculated by Euler's formula:
\[
	b_1 = b_0 - \chi + b_2
\]\[
	\chi = (\#\mbox{ of vertices}) + (\#\mbox{ of edges}) - (\#\mbox{ of faces})
\]
The second Betti number, $b_2$, is 1 if the subgraph contains the entire graph 
and the graph is the discretization of an orientable manifold; otherwise it is 
0. We define $b_i(s)$ to be the $i$th Betti number of the subgraph of all 
vertices, edges, and faces of state $s$.

During each experiment, we calculated the Betti numbers of each state every 1.0 
time. We sort the Betti numbers for states 1, 2, and 3 and concatenate them so 
that the resulting \textbf{pattern codes} are invariant under permutation of 
the non-empty states. For example, FIG.~\ref{fig:cleaning} is in state 111011: 
the 0th Betti number of all three states are 1, the 1st Betti number of states 
1 and 3 are 1, and the 1st Betti number of state 2 is 0.

Finally, in order to eliminate transitory fluctuations that represent flaws in 
our analytic procedure rather than genuine changes in the system, we 
reclassified any pattern that lasted for less than 5.0 time as ``transitory''. 
If the system entered a transitory state for 5.0 or less time and then returned
 to its original state, we removed that transitory fluctuation, reclassifying 
it as the surrounding pattern.

\section{Experiments}
\label{sec:experiments}

We explored the effect of graph topology on the critical mobility threshold of 
the $(3, 1)$ May-Leonard system using three graphs, as shown in 
FIG.~\ref{fig:topology}: a 2-dimensional periodic lattice, equivalent to a 
discretization of a torus; a 2-dimensional lattice with edges wrapped to form a
 discretization of a topological 2-sphere; and a 2-dimensional lattice with the
 edges wrapped to form a Klein bottle. These graphs are all locally equivalent,
 but have distinct global topologies.

We began by exploring feasible patterns that could occur on each graph. We 
experimented with random initializations, as used in \cite{Reichenbach} and 
\cite{nonperiodic}, but we found that we could achieve coexistence at higher 
mobility using a structured initialization. We initialized the lattice in state
 0, then placed one block of each state in random non-overlapping locations on 
it, as shown in FIG.~\ref{fig:initializations}. We used four choices of lattice
 width, 32, 64, 128, and 256 vertices, and sampled a variety of mobilities 
between $1\times10^{-4}$ and $8\times10^{-4}$. We ran 256 experiments on AWS 
EC2 for each combination of hyperparameters, running each experiment for 1000.0
 time.

We calculated the pattern codes of each experiment every 1.0 time, giving us the
 0th and 1st Betti numbers of the planar subgraphs consisting of all vertices 
of a particular state, sorted so as to be invariant to permutations of the non-
empty state. We ignored patterns before 200 time, as these might be transients 
during system warmup. We also ignored patterns with a zero 0th Betti number for
 a non-empty state, as these are patterns where no species has gone extinct 
yet, but where some species has become so rarefied that it is eliminated by the
 cleaning process. Recovery from this state is rare, and does not significantly
 contribute to survival time. We looked for patterns for which we had at least 
one example that persisted for more than 100.0 time when $m \geq 
2.5\times10^{-4}$. This gave us three patterns on the Klein bottle, two on the 
sphere, and three on the torus. Examples of each pattern are shown in 
FIG.~\ref{fig:patterns}.

We saved the model state at the start of periods where the system remains in the
 same pattern for at least 100.0 time. We randomly sampled from these saved 
states as initializations and ran 256 additional experiments with new random 
seeds for each combination of graph, pattern, and width greater than 32, and a 
selection of mobilities. We excluded width 32 as the models were so unstable we
 could not gather enough initializations. For mobilities with insufficient 
examples, we sampled additional saved states from neighboring mobilities to 
ensure we had at least 64 initializations. We refer to these as patterned 
initializations, since they are intended to induce the formation of a specific 
pattern at the start of the experiment. We ran the models to see how long each 
pattern would last before transitioning to a new pattern, discarding 
experiments where, despite the initialization, the desired pattern did not 
form.

\section{Results and Discussion}
\label{sec:discussion}

In FIG.~\ref{fig:extinction}, we show the extinction curves of our experiments 
with block initialization. Each curve gives the probability that, in a system 
beginning with block initialization, at least one species will have gone 
extinct by the end of the experiment at 1000.0 time. We observe that the 
critical mobility threshold of the sphere is about $2.8\times10^{-4}$ while the
 critical mobility threshold of the torus and the Klein bottle are both about 
$5.5\times10^{-4}$, implying that the large-scale graph structure does have a 
significant effect on the critical mobility threshold.

In FIG.~\ref{fig:median-durations}, we show the median duration of each long-
lived pattern from our experiments with patterned initializon, estimated using 
the Kaplan-Meier estimator \cite{kaplan-meier} \cite{lifelines}. We observe 
that each pattern has a distinct median survival time, and that the survival of
 the torus and Klein bottle graphs at higher mobilities is due to a specific 
pattern, 111111. We dub this pattern the \textbf{marching bands}, since it 
consists of three bands of different states ``marching'' across the lattice, 
state 1 chasing state 2 chasing state 3 chasing state 1. The marching bands 
pattern is extremely stable, likely due to the absence of spirals, since there 
is no point on the marching bands where all three states meet.

It is easy to show that topological constraints make it impossible for the 
marching bands to form on the sphere. In what follows, $H_k(X)$ is the $k$th 
simplicial homology group of the manifold $X$.

\begin{theorem}\label{nobands}Let $U_1, U_2, U_3 \subset S^2$ be connected, non-
empty, open sets such that $\overline{U_1}\cup\overline{U_2}\cup\overline{U_3} 
= S^2$ and $U_i \cap U_j$ is empty for $i\neq j$. There exists no choice of 
$U_i$ such that $H_1(U_1) \cong H_1(U_2) \cong \mathbb{Z}$.\end{theorem}

	\begin{proof}Since $U_i$ are connected surfaces that are subsets of $S^2$, they
 must be orientable, possibly punctured, genus-0 surfaces. Therefore, they are 
classified by their first homology group, so $U_1$ must be a neighborhood of a 
circle, and $\partial U_1$ must be a disjoint pair of circles.

Let $\widetilde{U}_i$ be arbitrarily small neighborhoods of $U_i$, and let 
$V=\widetilde{U}_2\cup\widetilde{U}_3$. Then $V\cap\widetilde{U}_1$ is an 
arbitrarily small neighborhood of $\partial U_1$, and must have homology groups
 $H_0(V\cap\widetilde{U}_1)\cong H_1(V\cap\widetilde{U}_1) \cong \mathbb{Z}^2, 
H_2(V\cap\widetilde{U}_1)\cong \varnothing$.
		
By the Mayer-Veitoris Theorem, there is a pair of long exact sequences:
\[
	\varnothing \to H_2(S^2) \to H_1(V\cap\widetilde{U}_1) \to 
\]\[
	H_1(V) \oplus H_1(\widetilde{U}_1) \to H_1(S^2) \to ...
\]\[
	...\to H_1(S^2) \to H_0(V\cap\widetilde{U}_1) \to 
\]\[
	H_0(V) \oplus H_0(\widetilde{U}_1) \to H_0(S^2) \to \varnothing
\]
Which is equivalent to:
\[
	\varnothing \to \mathbb{Z} \to \mathbb{Z}^2\to H_1(V)\oplus\mathbb{Z} \to 
\varnothing
\]\[
	\varnothing \to \mathbb{Z}^2 \to H_0(V) \oplus \mathbb{Z} \to \mathbb{Z} \to 
\varnothing
\]
From which we conclude that $H_1(V) \cong \varnothing, H_0(V)\cong 
\mathbb{Z}^2$. Therefore, $V$ has two components, both of which are orientable 
open surfaces with genus zero and no punctures. Since $\widetilde{U}_2, 
\widetilde{U}_3$ are both connected, each must correspond to one of those two 
components, and so $H_1(U_2)\cong H_1(U_3)\cong \varnothing$.

\end{proof}

\begin{corollary}The marching bands cannot form on a spherical 
graph.\end{corollary}

\begin{proof} This is immediate from Theorem~\ref{nobands}.\end{proof}

However, the existence of the marching bands is not the sole distinction between
 the graphs. We find that if we eliminate the marching bands, the spherical 
graph actually has a slightly higher critical mobility. In 
FIG.~\ref{fig:extinction-without-marching-bands}, we plot the probability that 
at least one species goes extinct by 1000.0 time for our experiments using 
block initialization, excluding experiments that entered the marching bands 
pattern at any time. We observe that, excluding experiments that entered the 
marching bands, the critical mobility threshold of the periodic lattice and the
 Klein bottle are both $2.5\times10^{-4}$, compared to $2.8\times10^{-4}$. This
 implies that the effect of graph topology is not limited to the ability of the
 torus and Klein bottle to support patterns without spirals, because even in 
spiraling patterns, the critical mobility threshold of the sphere is slightly 
different.

Interestingly, the non-orientability of the Klein bottle appears to have no 
impact on the critical mobility threshold. The Klein bottle has higher 
probability of extinction for a given mobility, but the critical mobility 
thresholds are identical. We attribute this to the fact that the marching bands
 can only form in one direction on the Klein bottle, as opposed to two on the 
torus, making it less likely this pattern will form during model warmup.

\section{Conclusions and Further Research}
\label{sec:conclusion}

In this work we have shown that the critical mobility threshold of the $(3, 1)$ 
May-Leonard system is a property of the global graph structure, due to 
topological constraints on which patterns can form on which graphs.

Future research will focus on a more detailed understanding of the role topology
 plays in what patterns can form and how they evolve. The Betti number 
calculation we used is tractable only for discretizations of 2-dimensional 
manifolds; finding a more general homological calculation would allow us to 
extend this analysis to other graphs.


%

\section{Figures}

\begin{figure}
	\begin{center}
	\includegraphics[width=4.7cm]{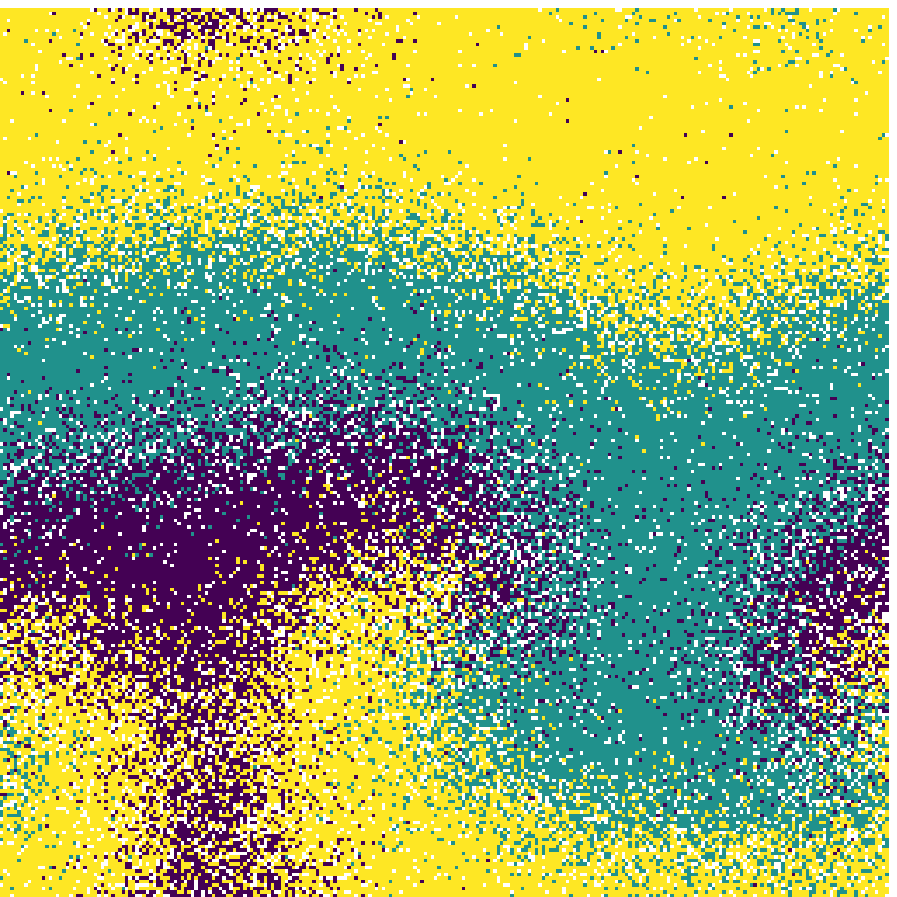}
	\includegraphics[width=4.7cm]{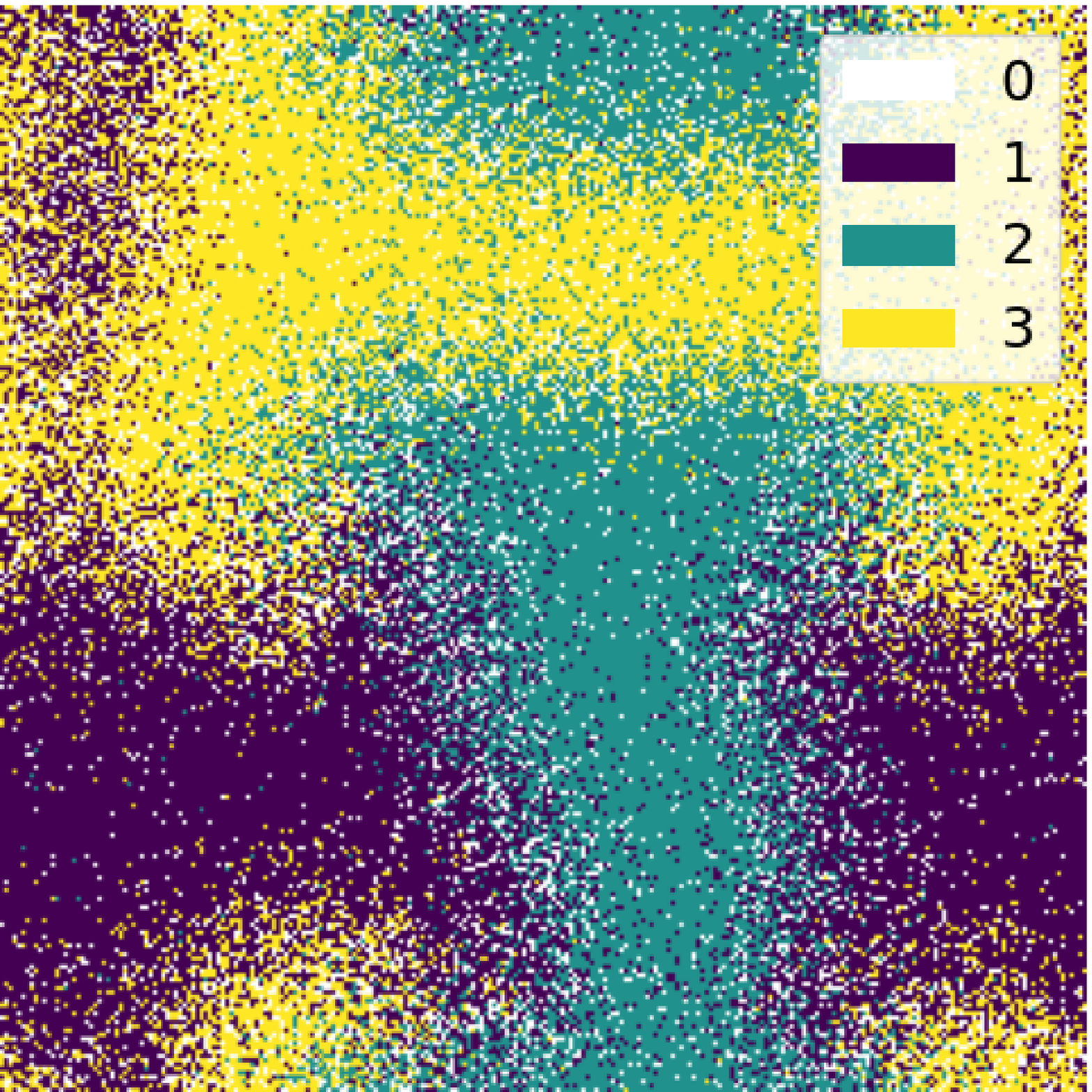}
	\end{center}
	\caption{Example model snapshots from the $(3, 1)$ May-Leonard system 
demonstrating different patterns
		\label{fig:example}
	}
\end{figure}

\begin{figure}
	\begin{center}
	\includegraphics[width=4.7cm]{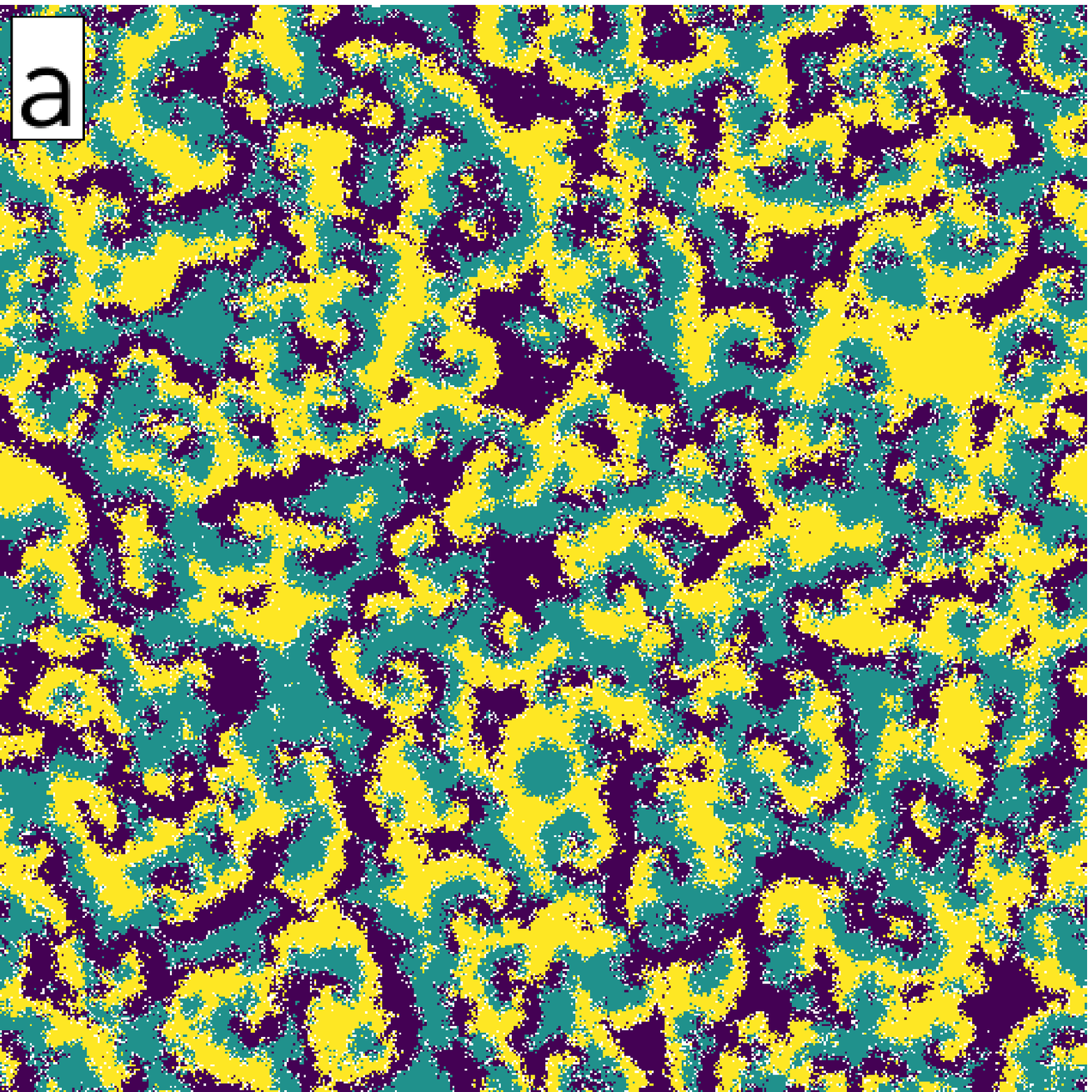}
	\includegraphics[width=4.7cm]{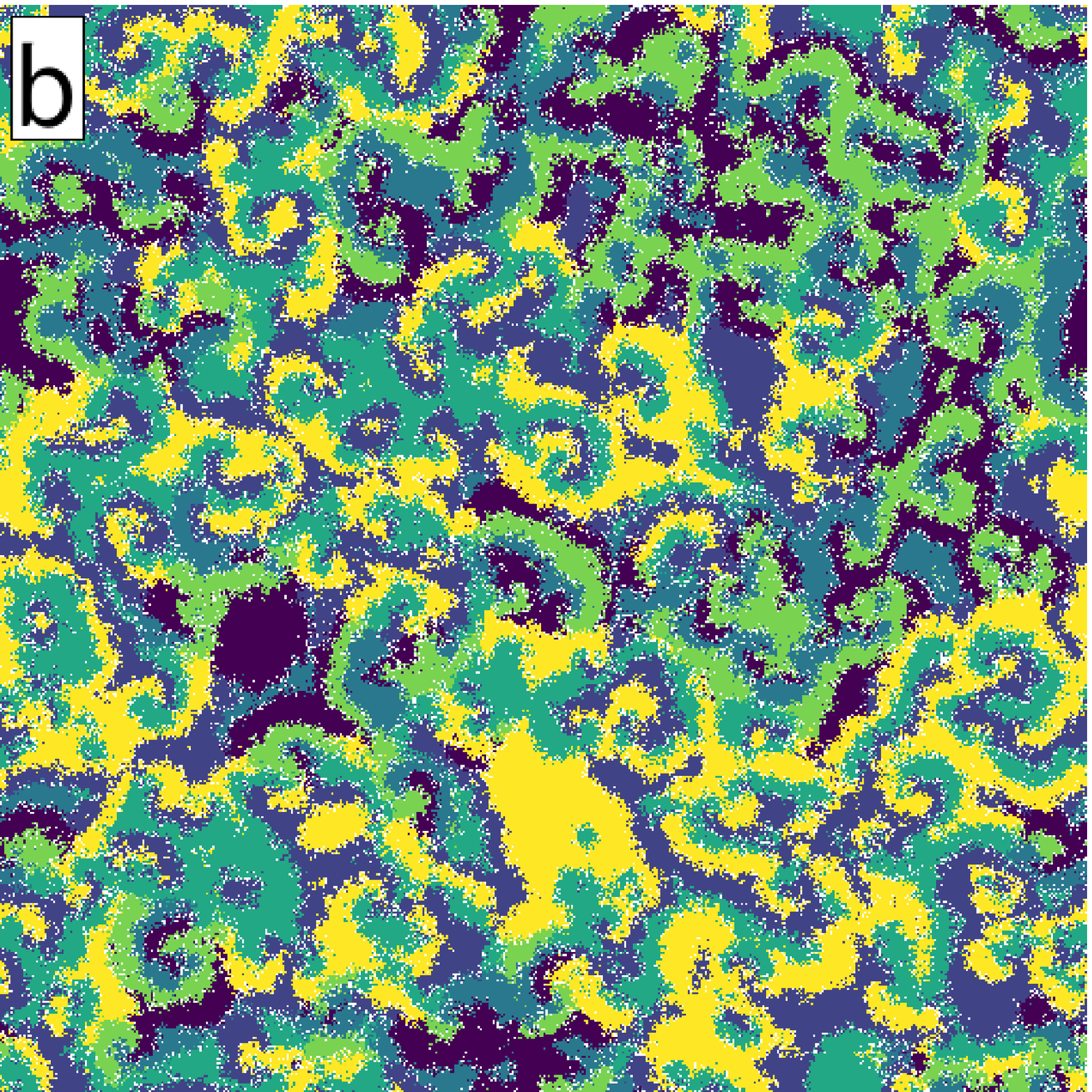}
	\\
	\includegraphics[width=4.7cm]{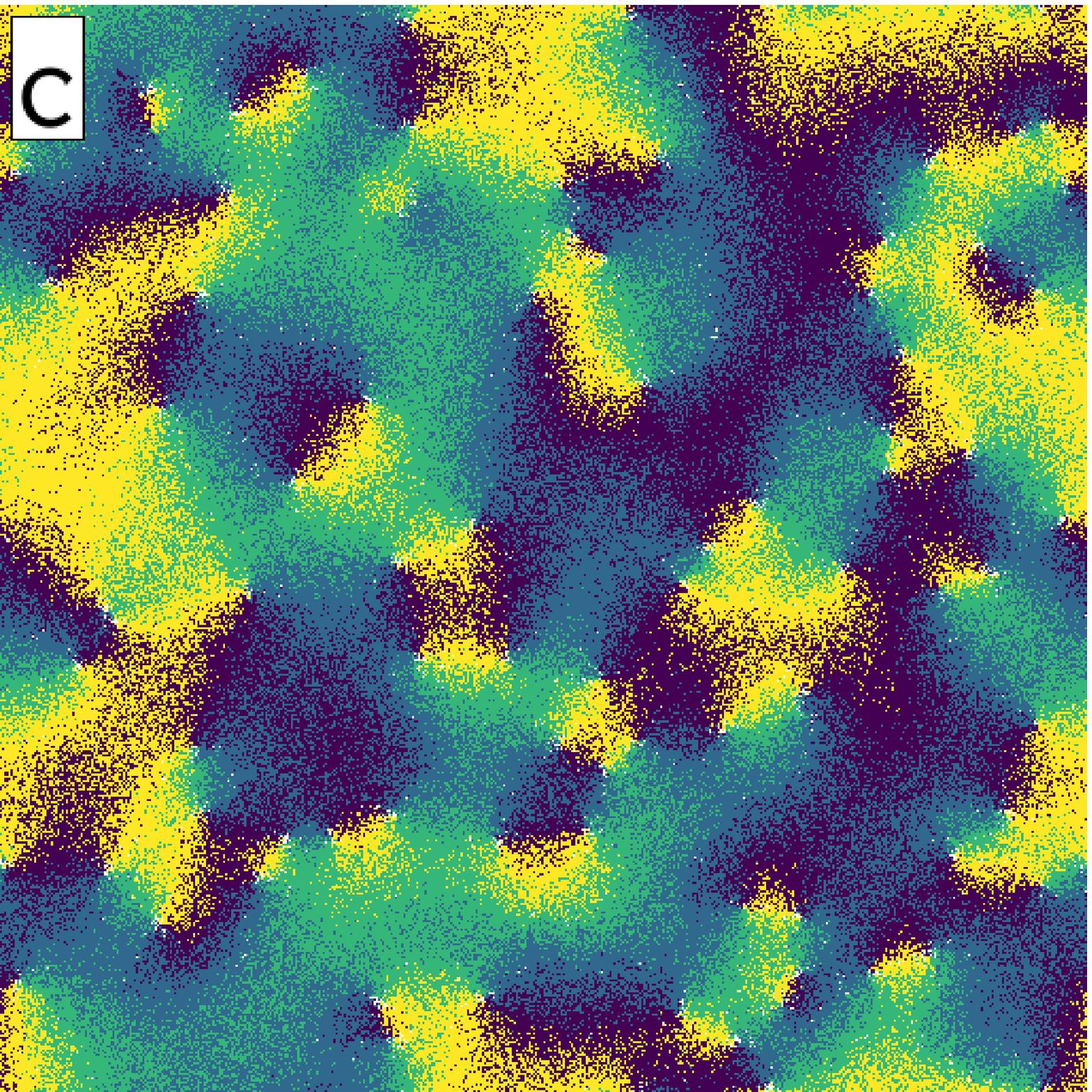}
	\includegraphics[width=4.7cm]{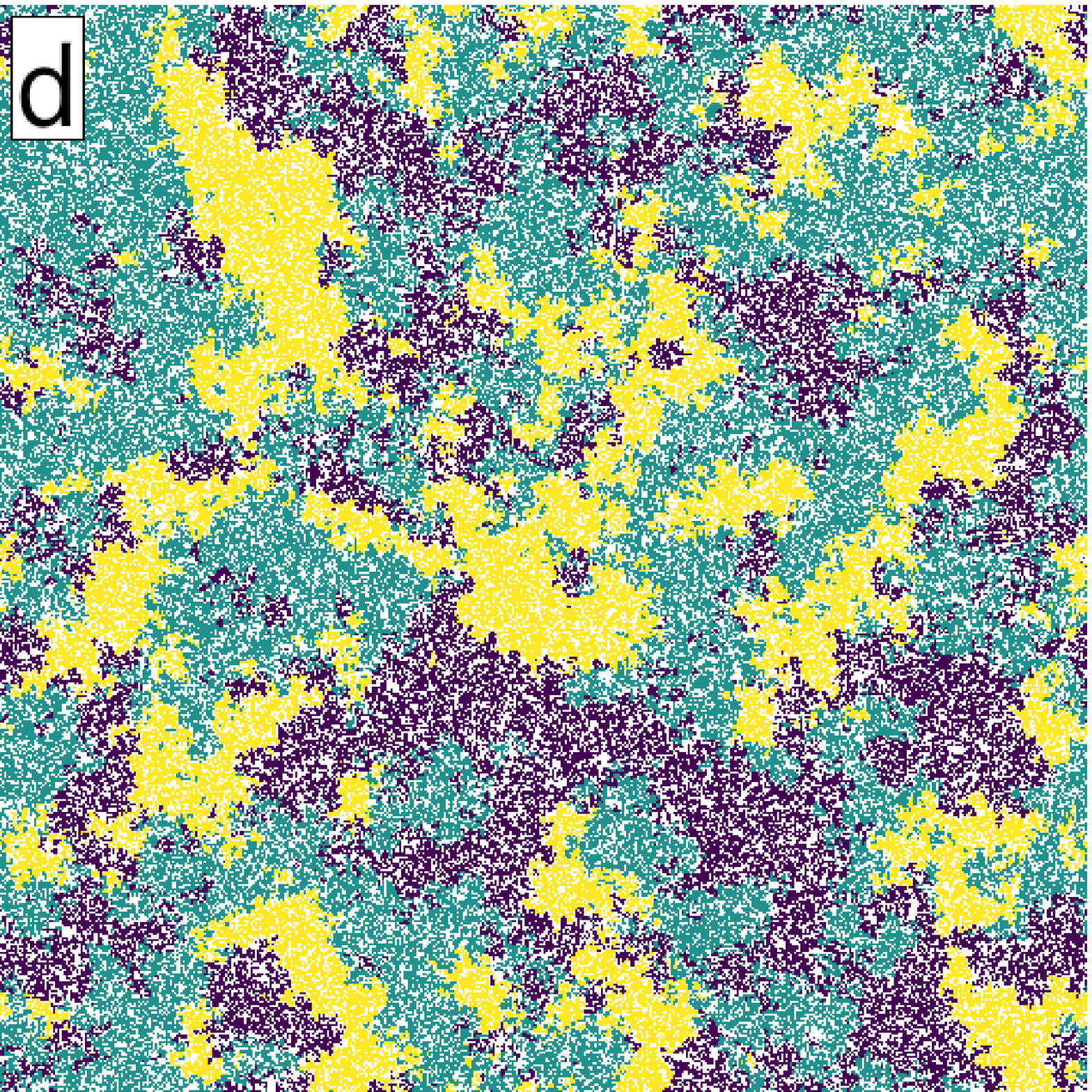}
	\end{center}
	\caption{Samples from Viridicle runs of a: $(3, 1)$ May-Leonard model 
\cite{nr}, b: $(6, 3)$ May-Leonard model \cite{nr}, c: $Z_N$ Lotka-Volterra 
competition \cite{strings}, and d: Durret \& Levin \emph{E. Coli} model 
\cite{ecoli}
		\label{fig:examples}
	}
\end{figure}

\begin{figure}
	\begin{center}
	\includegraphics[width=4.7cm]{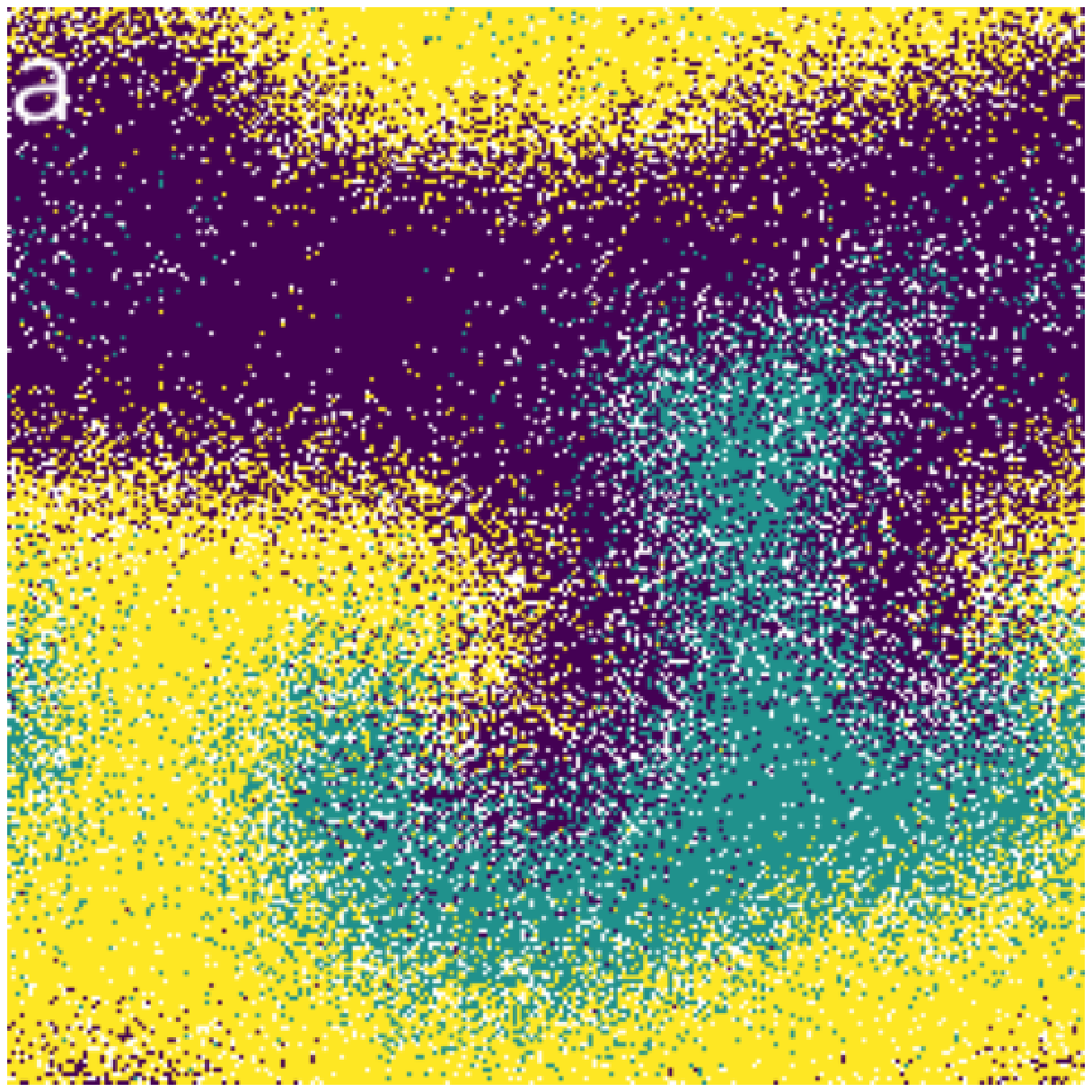}
	\includegraphics[width=4.7cm]{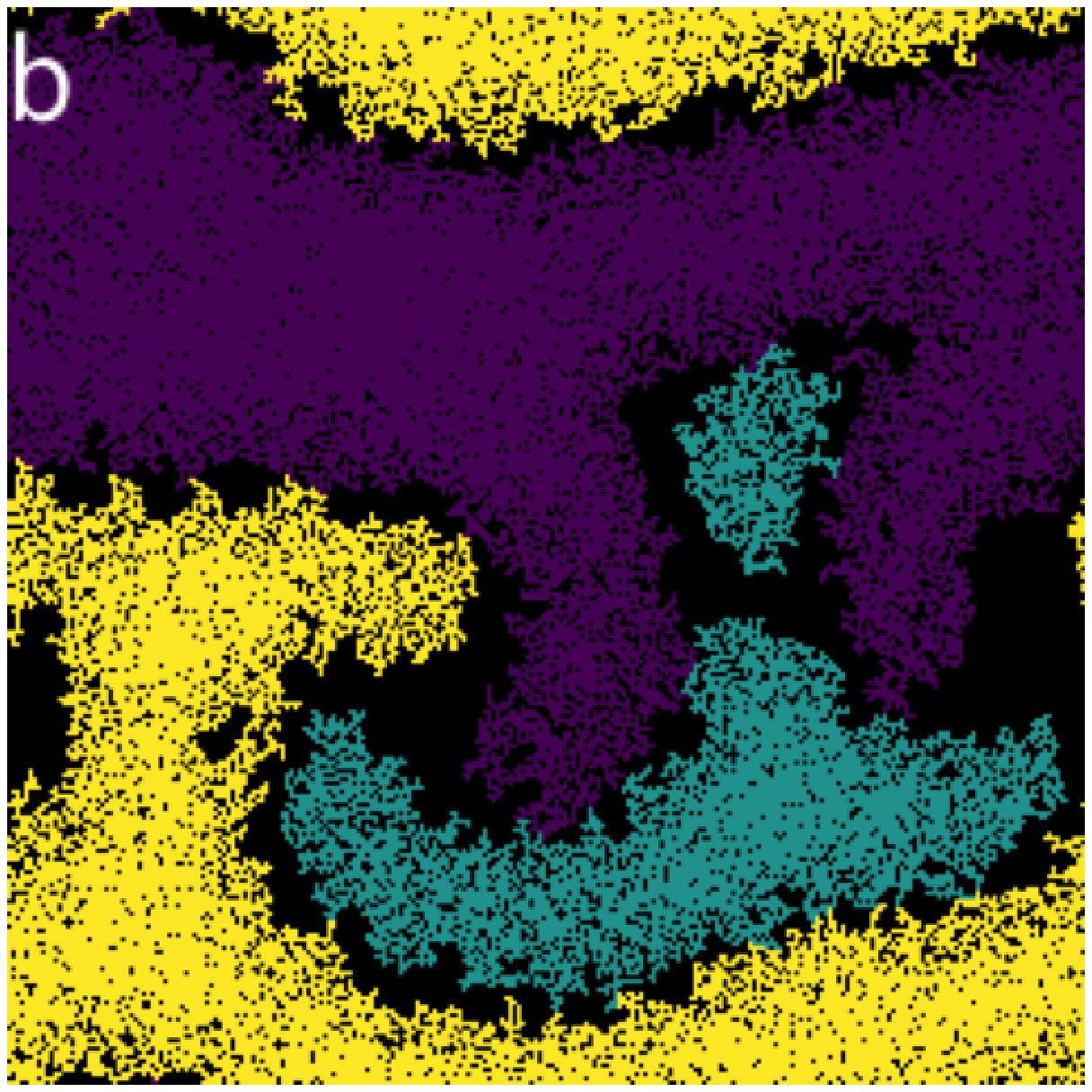}
	\includegraphics[width=4.7cm]{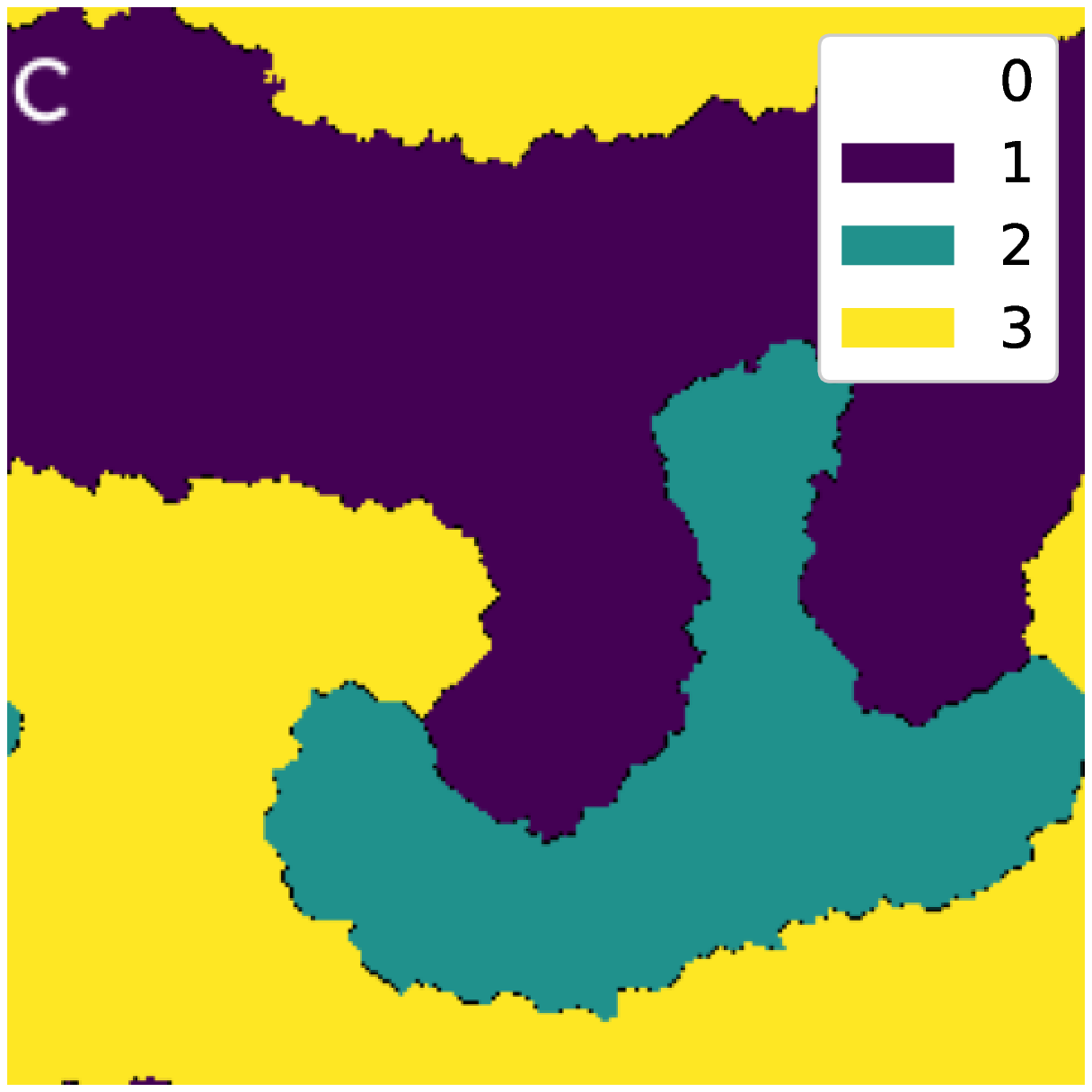}
	\end{center}
	\caption{a: Model prior to cleaning; b: Clusters below 128 size are set to 255;
 c: Clusters are expanded to fill remaining gaps
		\label{fig:cleaning}
	}
\end{figure}

\begin{figure}
	\begin{center}
	\includegraphics[width=4.7cm]{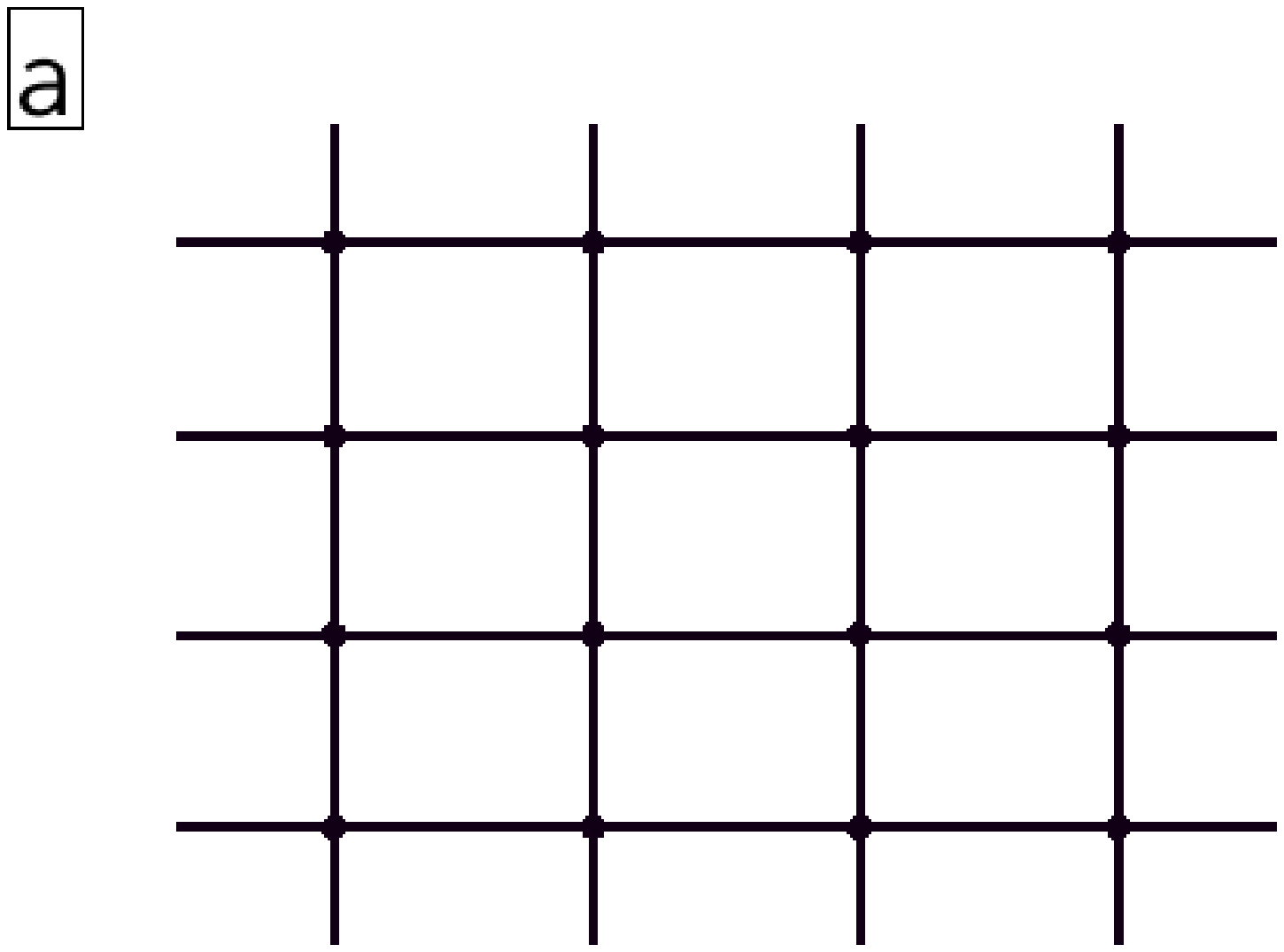}\hfill
	\includegraphics[width=4.7cm]{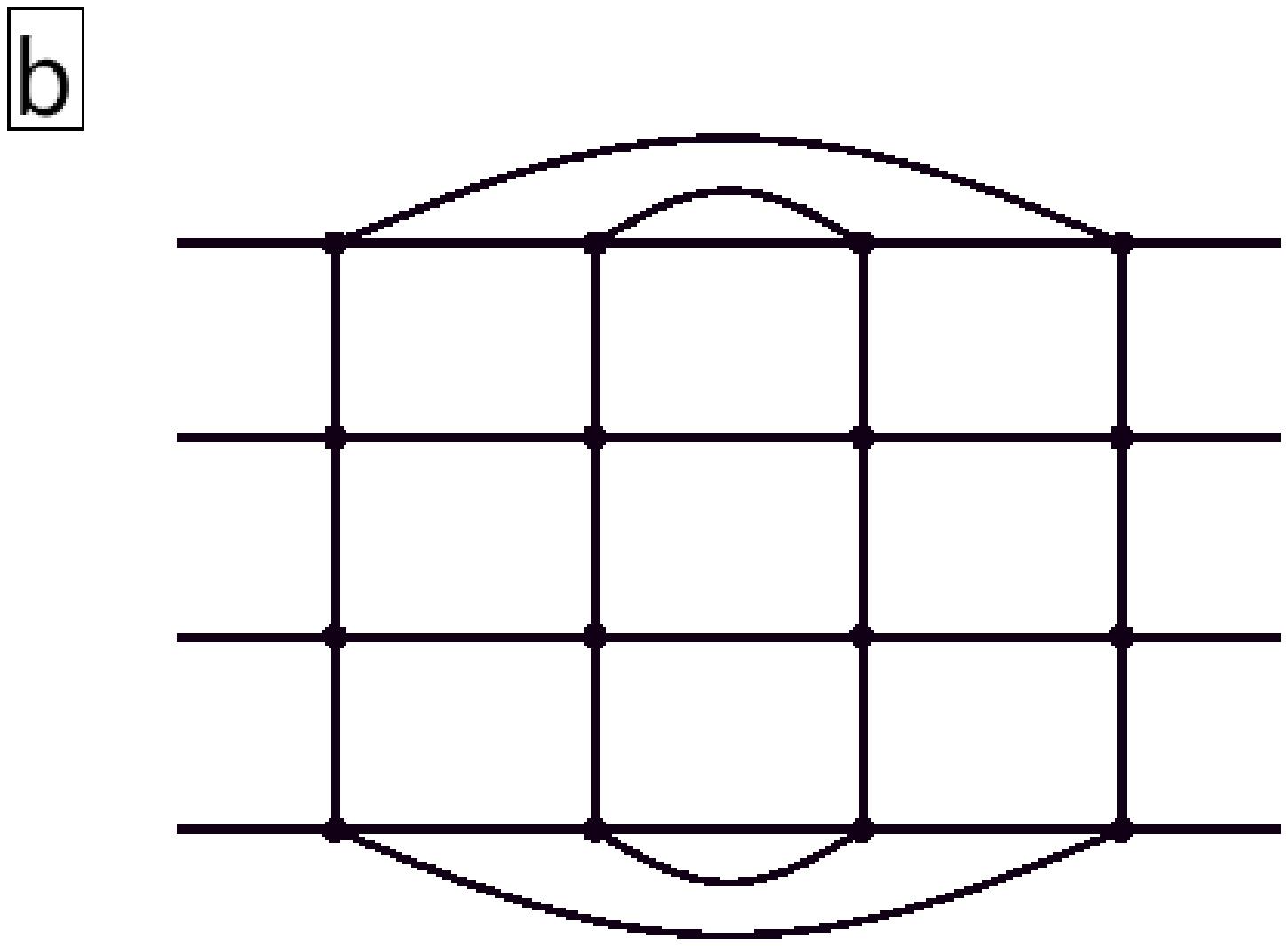}\hfill
	\includegraphics[width=4.7cm]{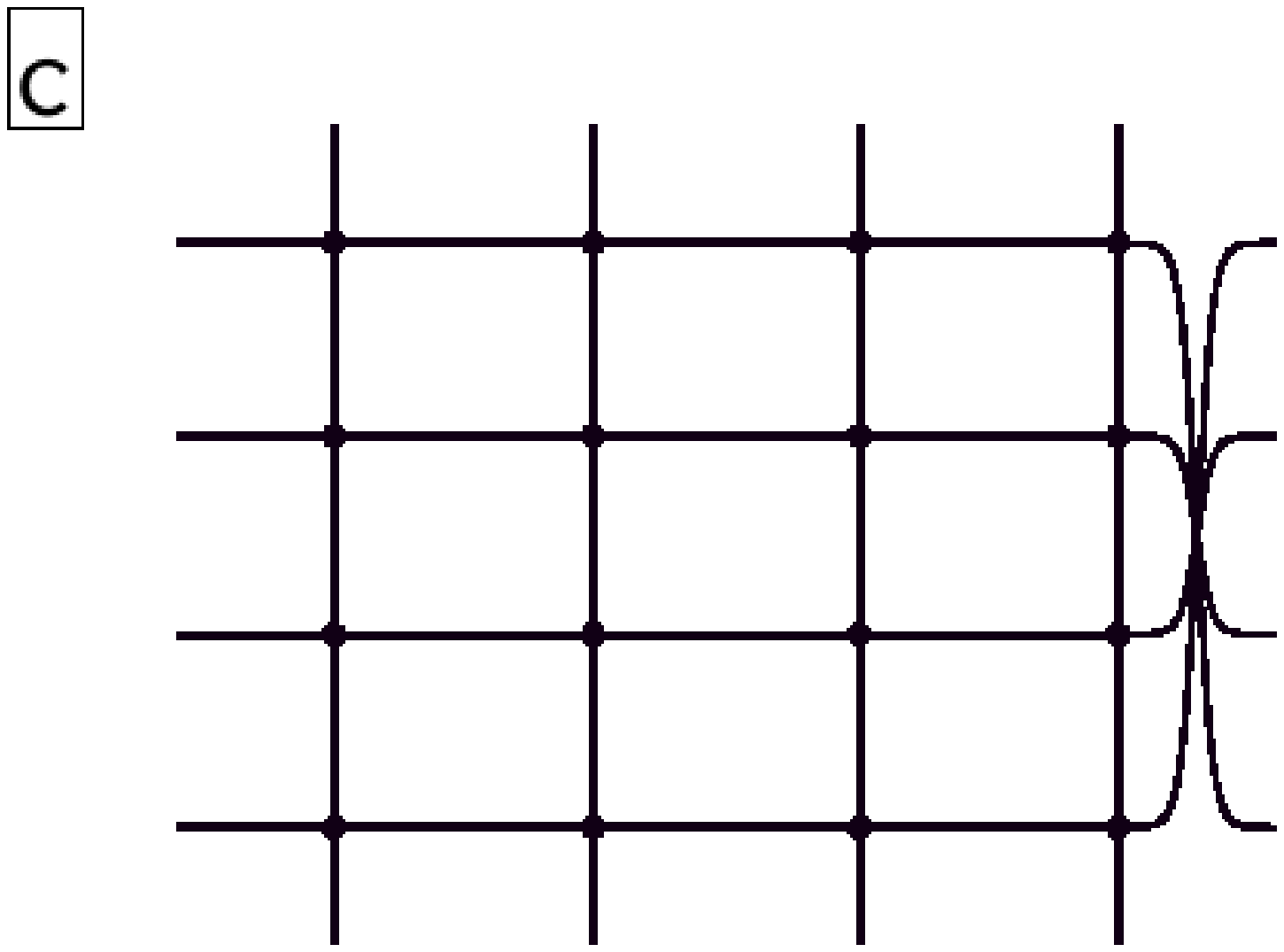}
	\end{center}
	\caption{a: $4\times4$ periodic lattice, b: sphere, and c: Klein bottle. The 
edges of each image are periodic.
		\label{fig:topology}
	}
\end{figure}

\begin{figure}
	\begin{center}
	\includegraphics[width=4.7cm]{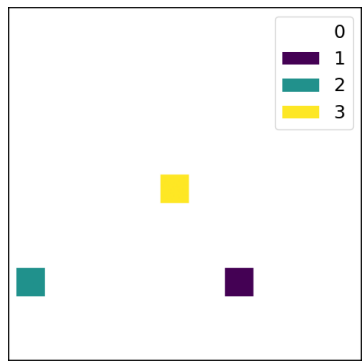}
	\end{center}
	\caption{Example of block initialization - blocks are placed randomly on the 
lattice.
		\label{fig:initializations}
	}
\end{figure}

\begin{figure}
	\begin{center}
	\includegraphics[width=4.7cm]{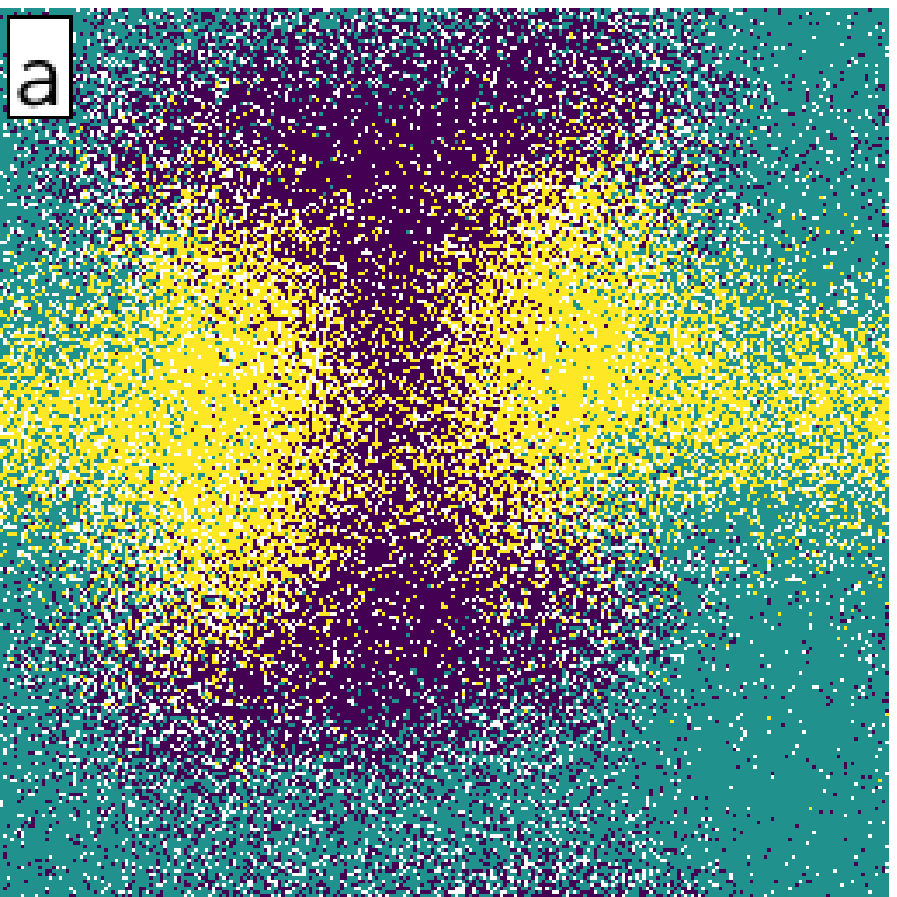}\hfill
	\includegraphics[width=4.7cm]{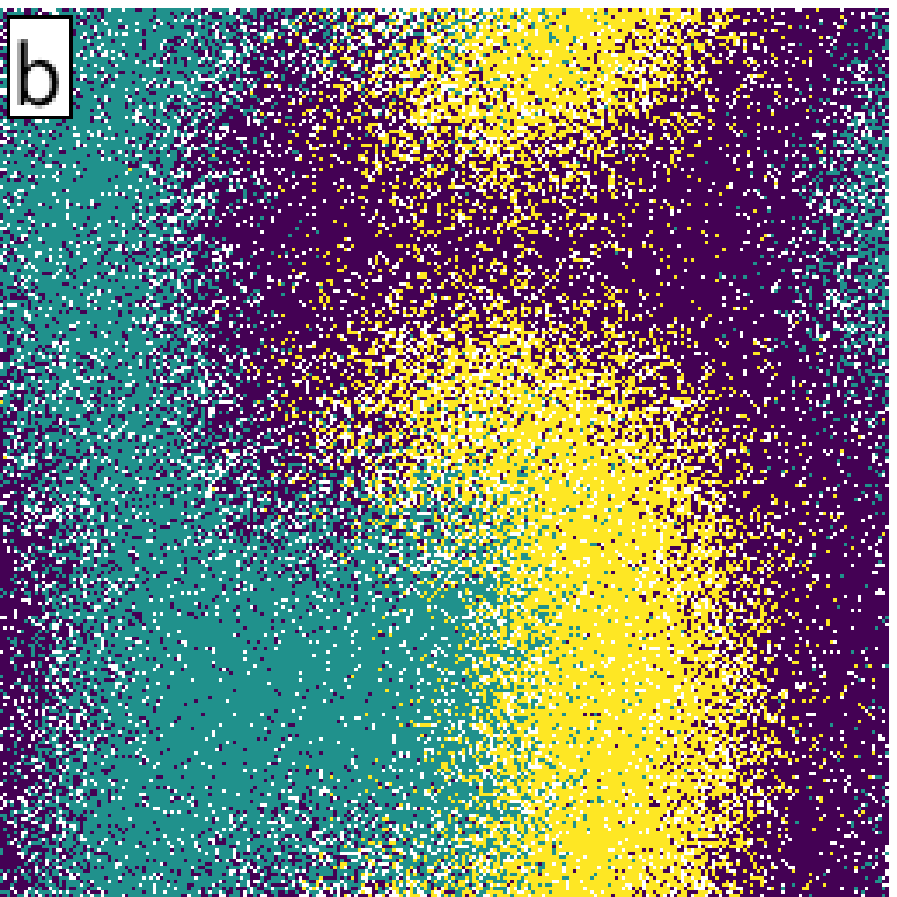}\hfill
	\includegraphics[width=4.7cm]{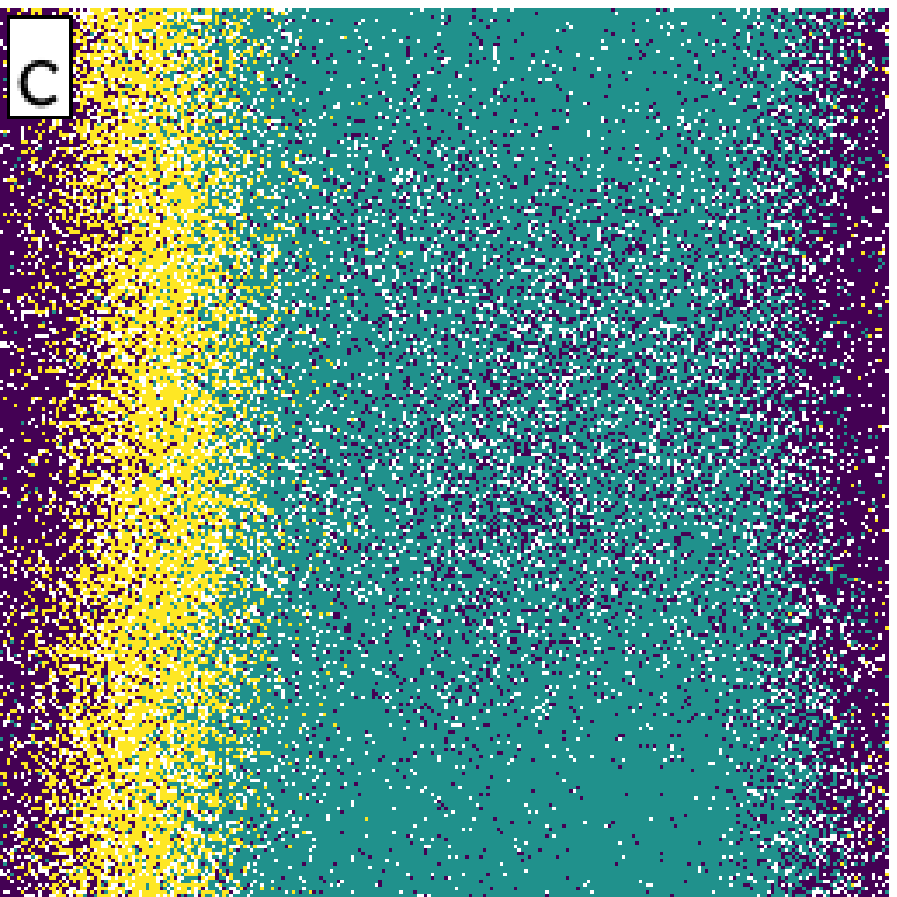}
	\\[\smallskipamount]
	\includegraphics[width=4.7cm]{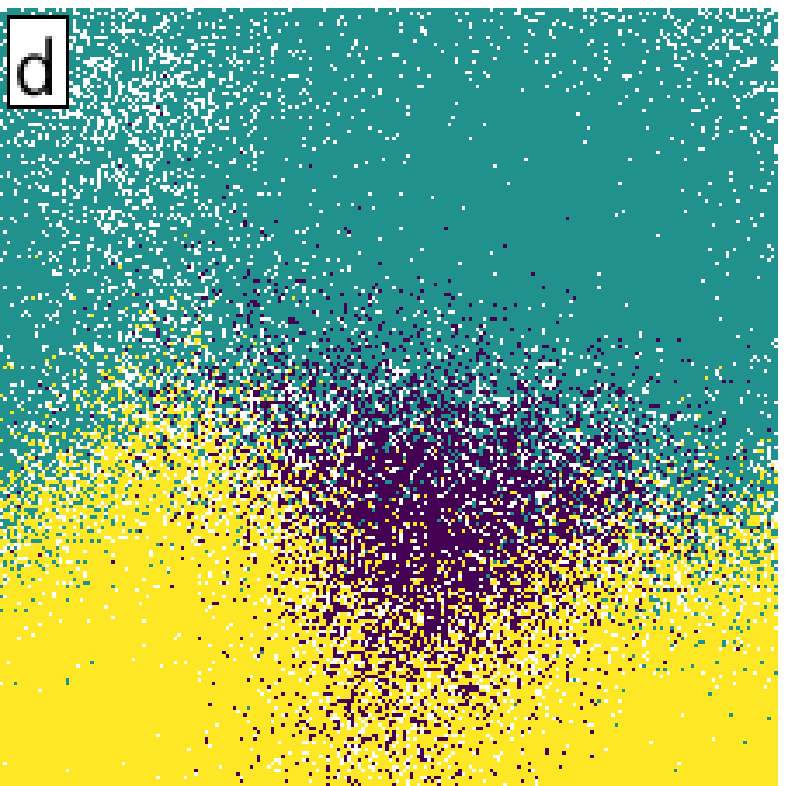}\hfill
	\includegraphics[width=4.7cm]{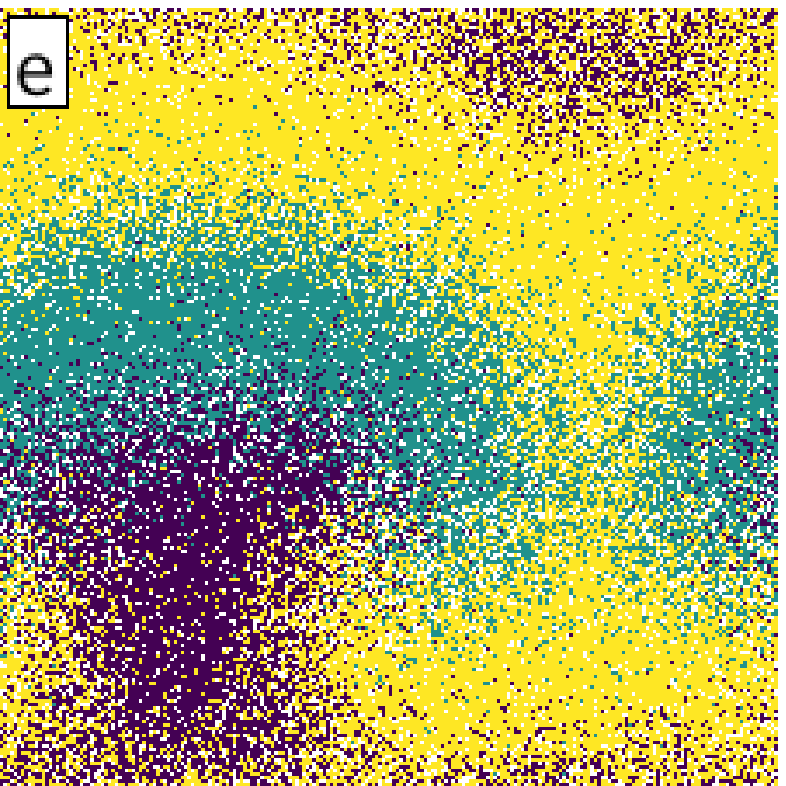}
	\\[\smallskipamount]
	\includegraphics[width=4.7cm]{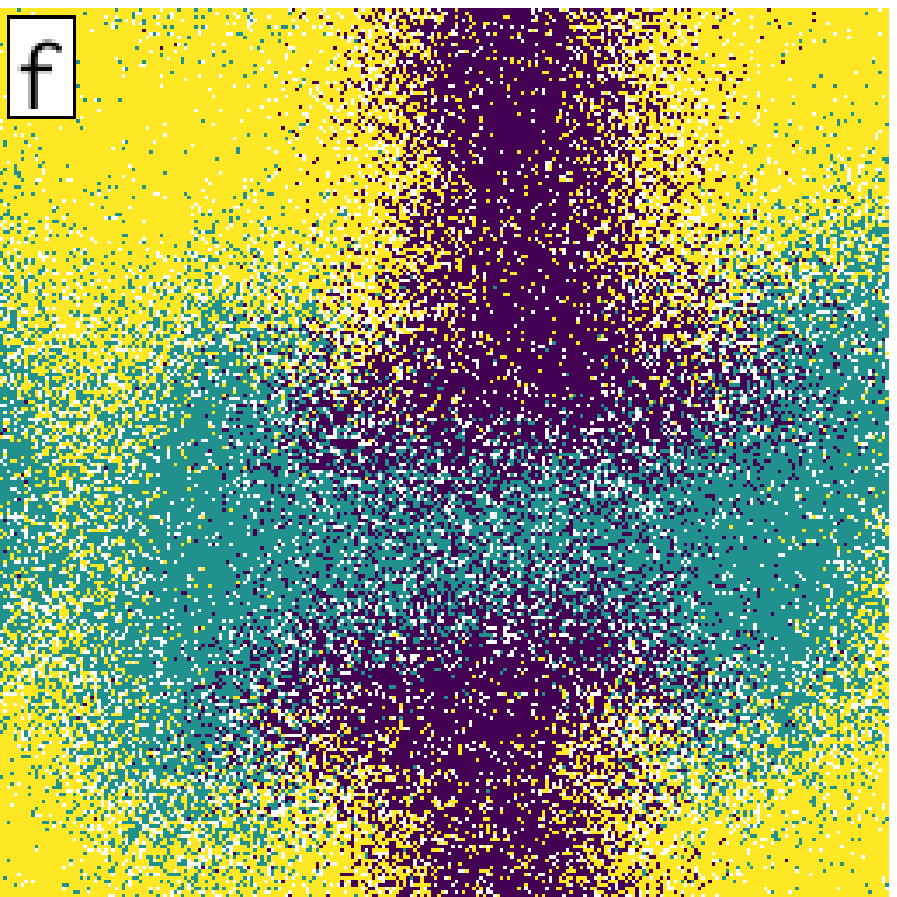}\hfill
	\includegraphics[width=4.7cm]{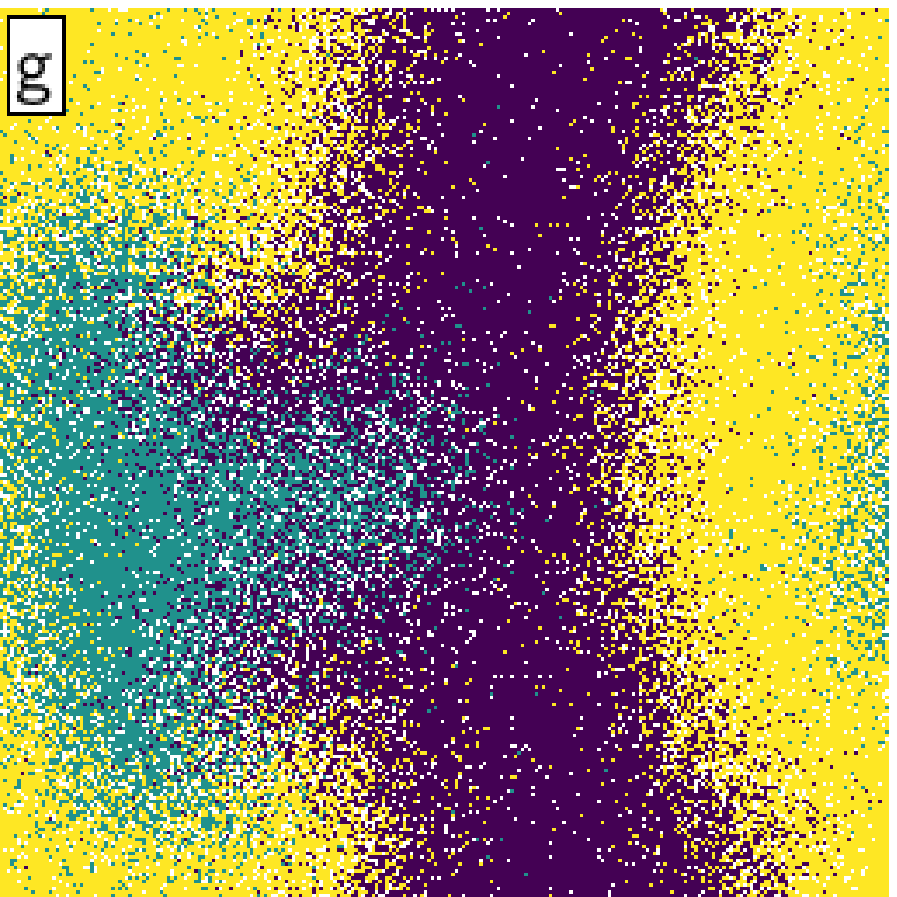}\hfill
	\includegraphics[width=4.7cm]{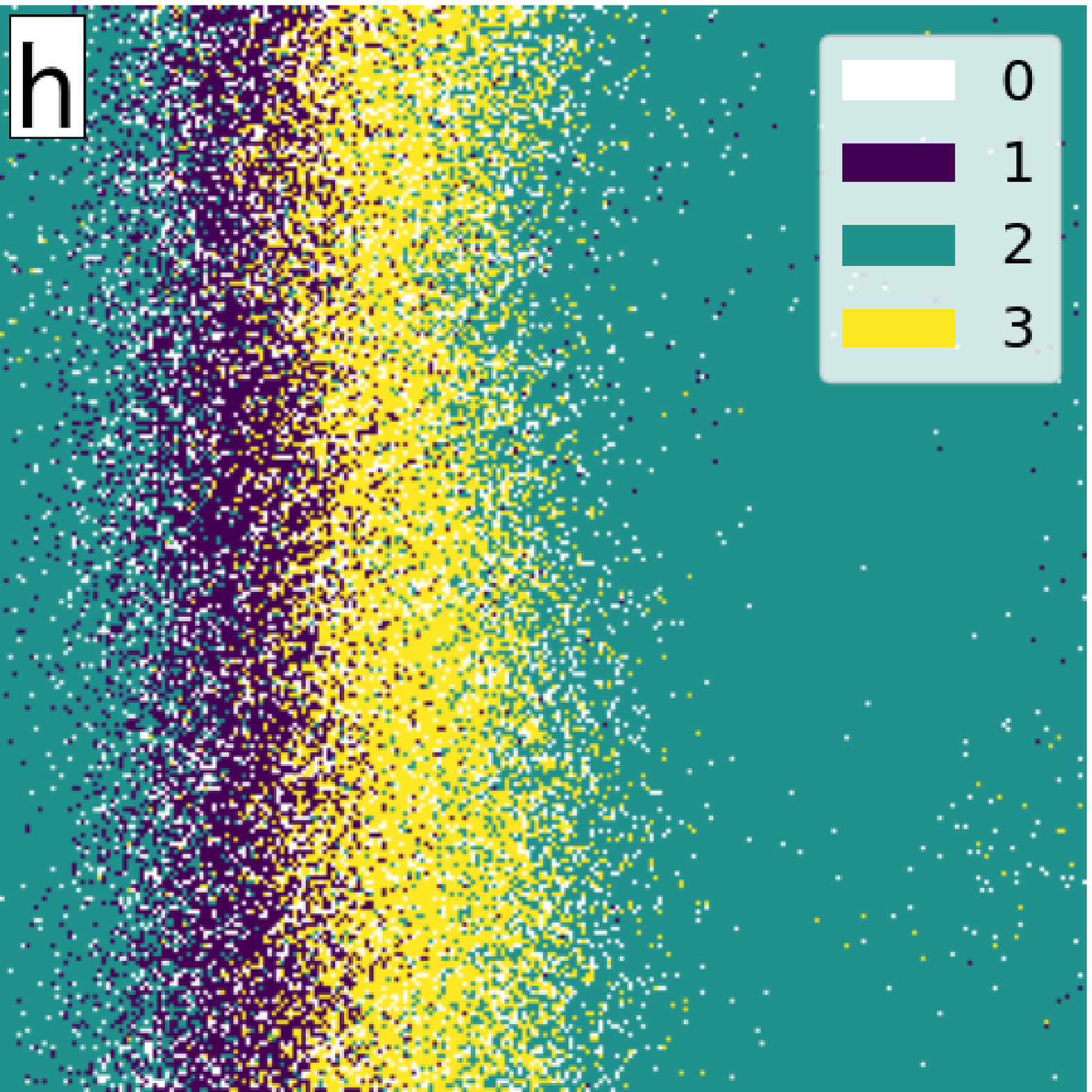}
	\end{center}
	\caption{Top row: lattice. a: Pattern code 111001, b: 111011, c: 111111. Middle
 row: sphere. d: 111000, e: 121001. Bottom row: torus. f: 111001, g: 111011, h:
 111111. 111111 are the ``marching bands.''
		\label{fig:patterns}
	}
\end{figure}

\begin{figure}
	\begin{center}
	\includegraphics[width=4.7cm]{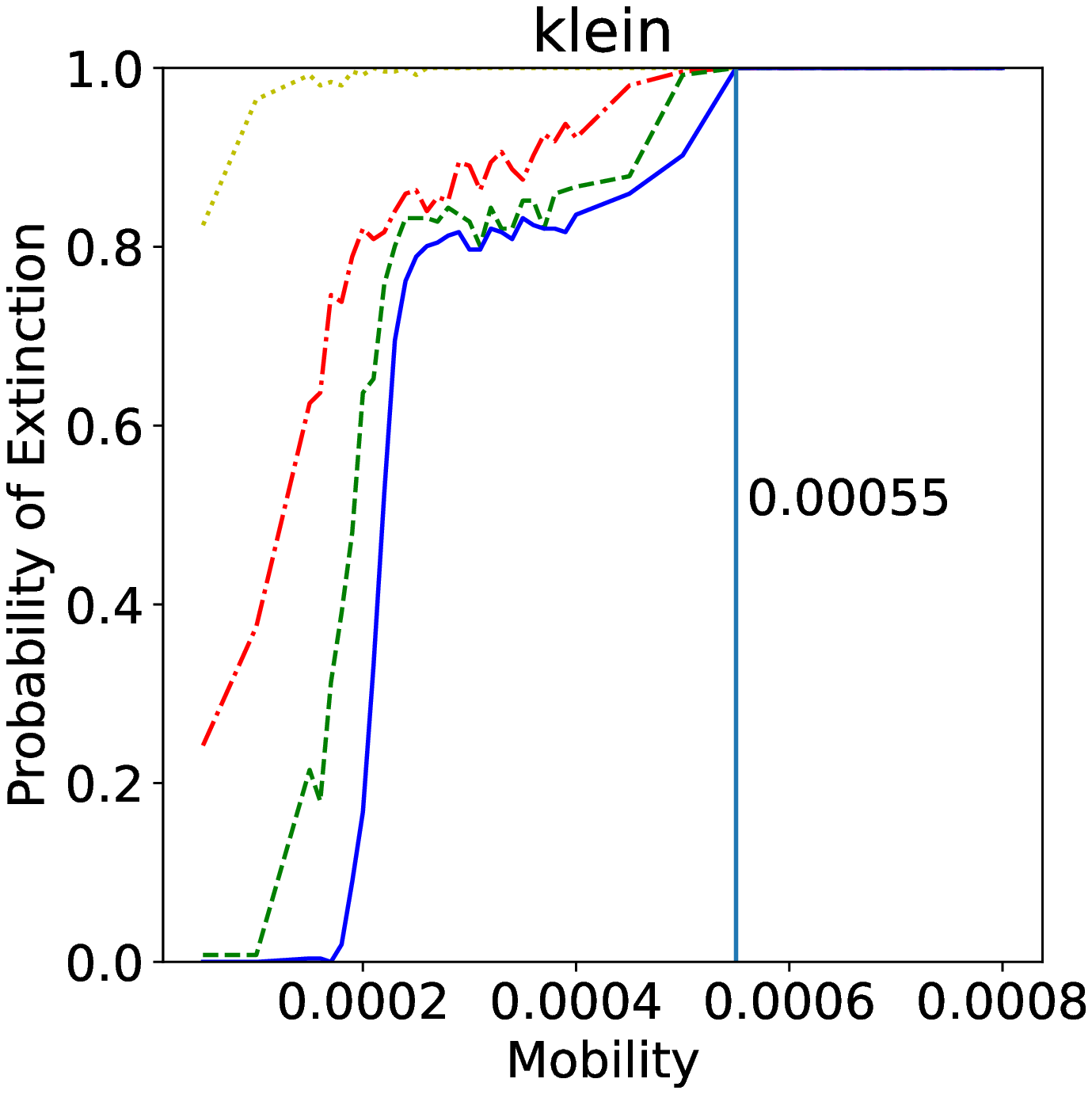}\hfill
	\includegraphics[width=4.7cm]{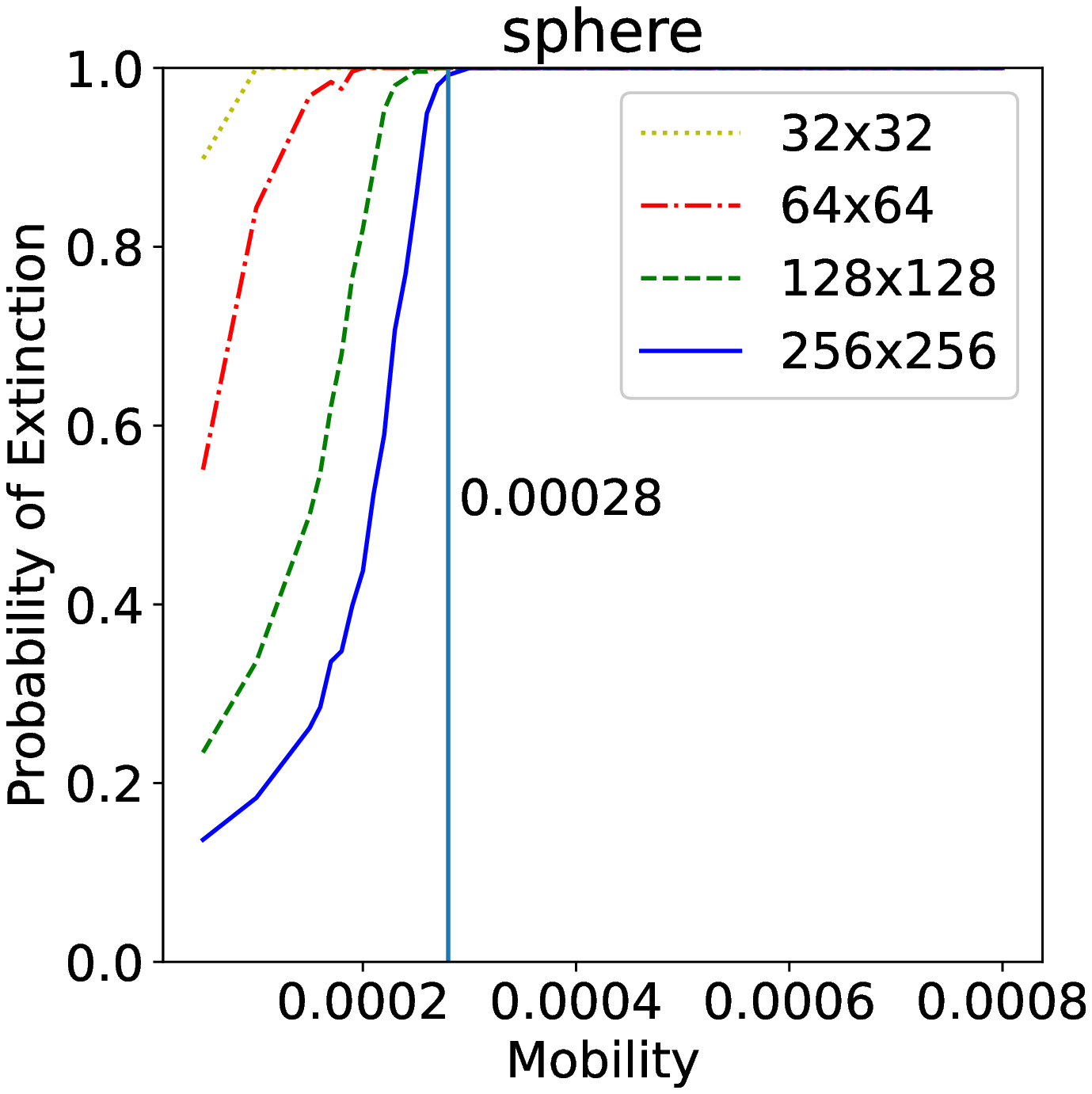}\hfill
	\includegraphics[width=4.7cm]{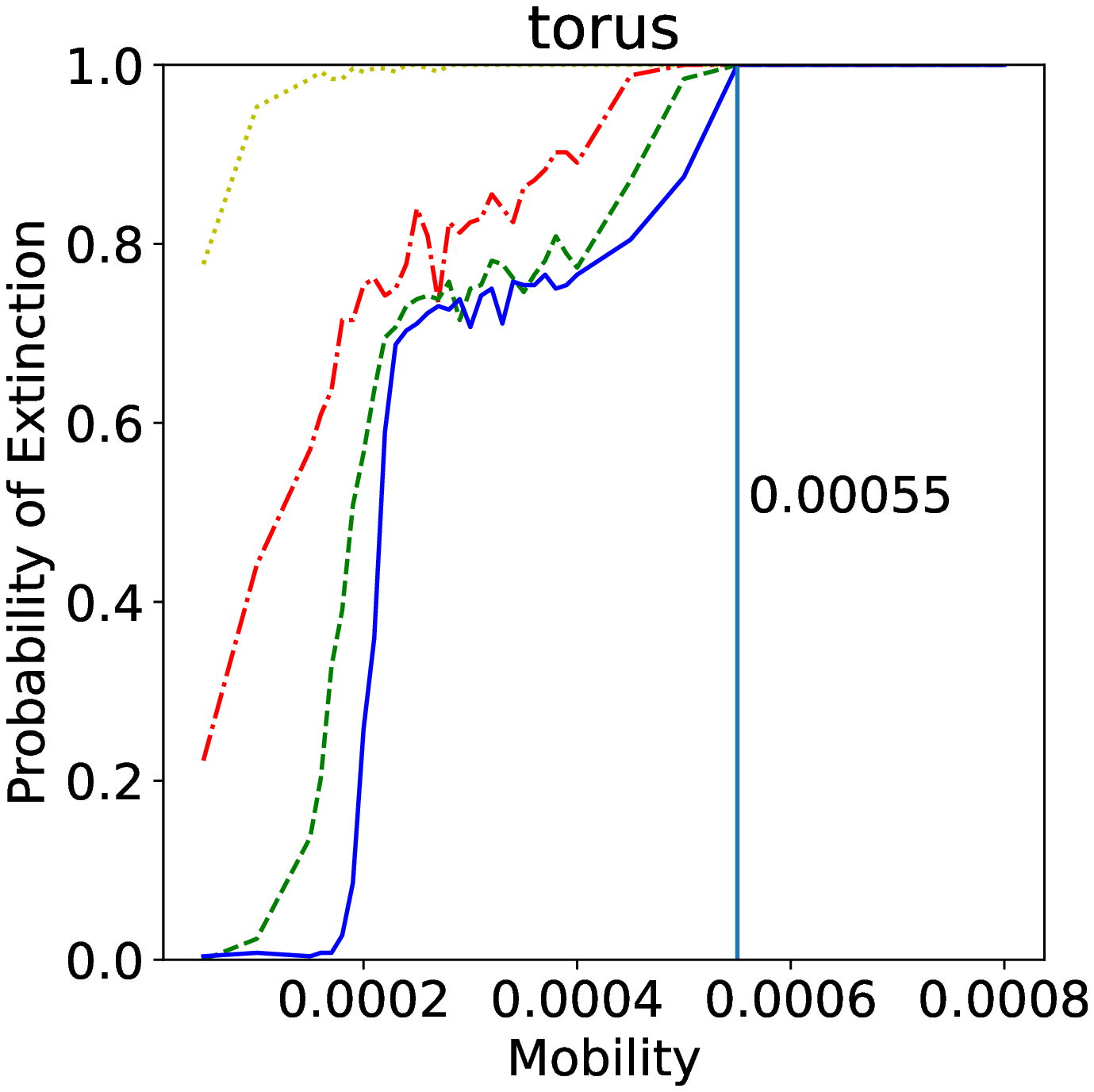}
	\end{center}
	\caption{Probability of extinction by 1000 time versus mobility, from 
experiments with block initialization. Each curve represents a different 
lattice size. Vertical lines are drawn at the mobility threshold where $99\%$ 
of all experiments see at least one species go extinct by 1000 time.
		\label{fig:extinction}
	}
\end{figure}

\begin{figure}
	\begin{center}
	\includegraphics[width=4.7cm]{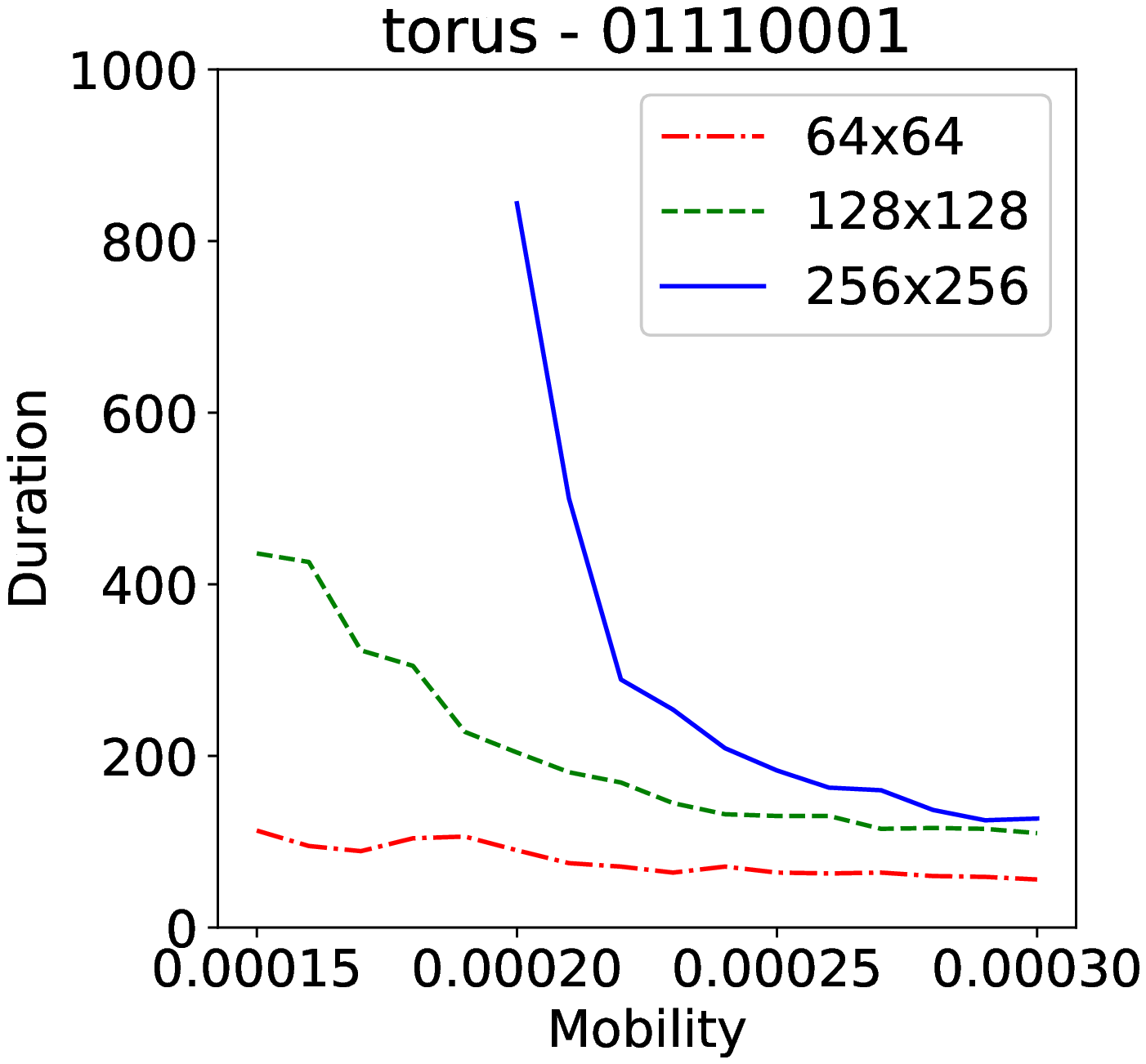}\hfill
	\includegraphics[width=4.7cm]{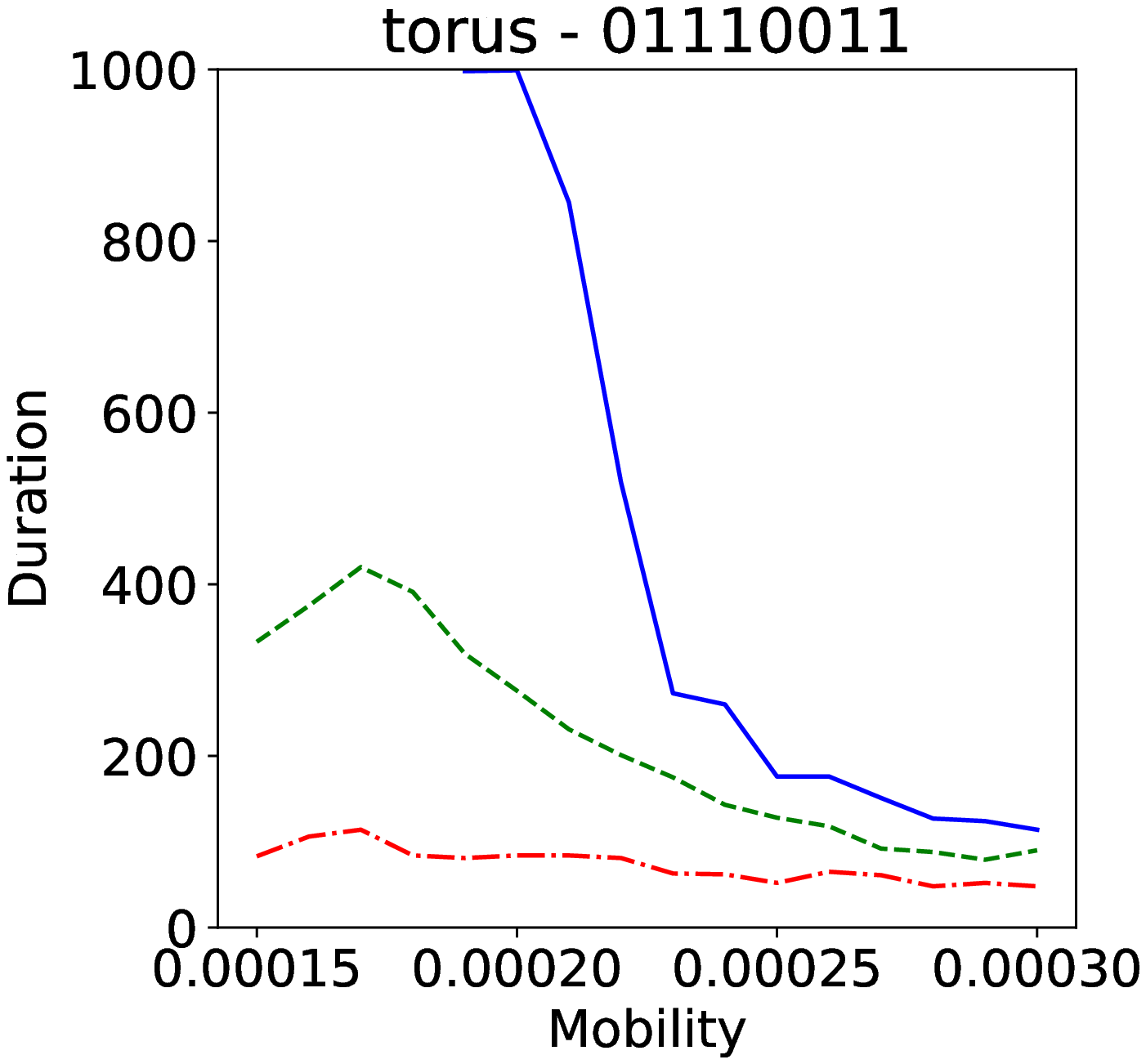}\hfill
	\includegraphics[width=4.7cm]{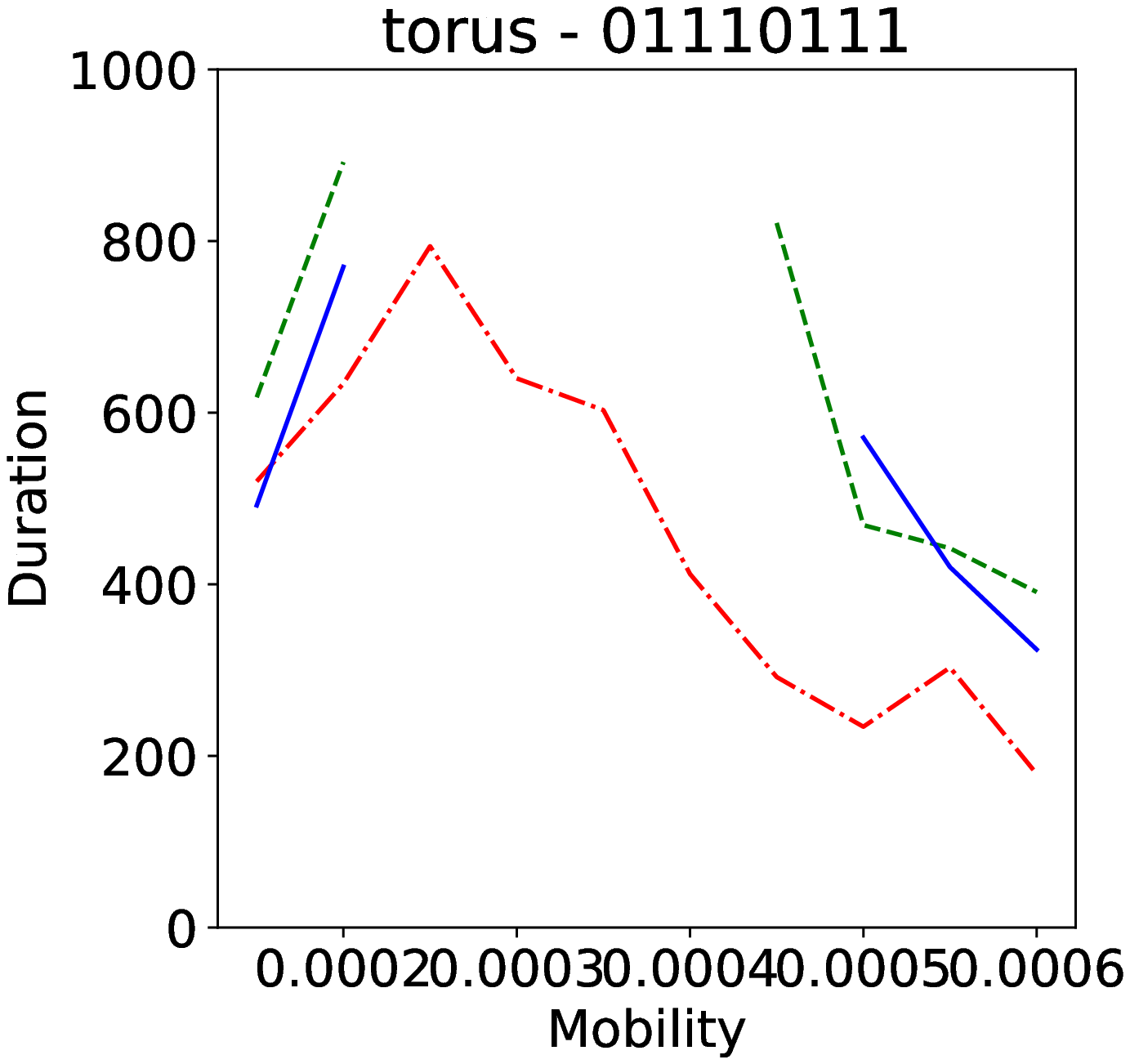}
	\\[\smallskipamount]
	\includegraphics[width=4.7cm]{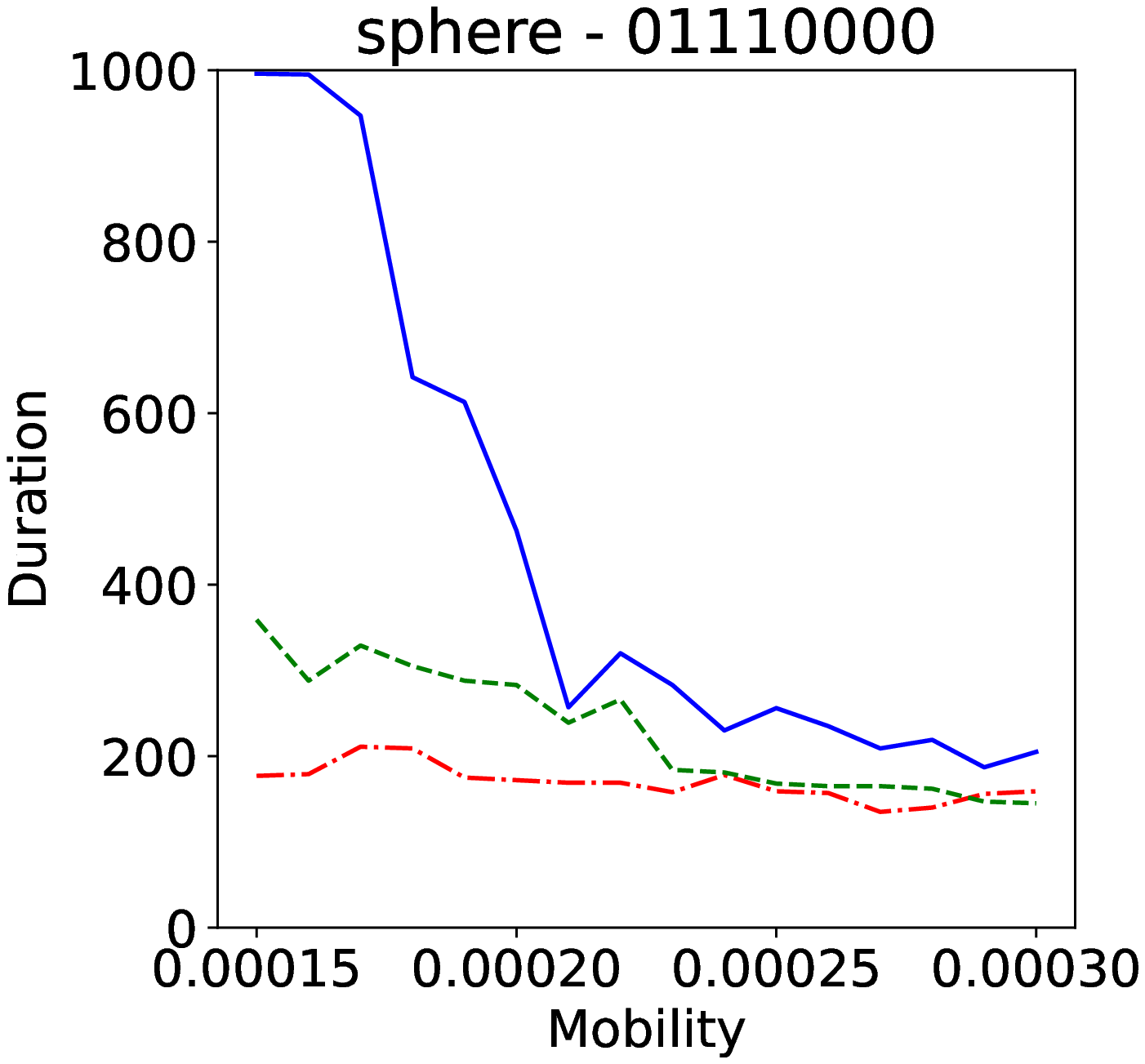}\hfill
	\includegraphics[width=4.7cm]{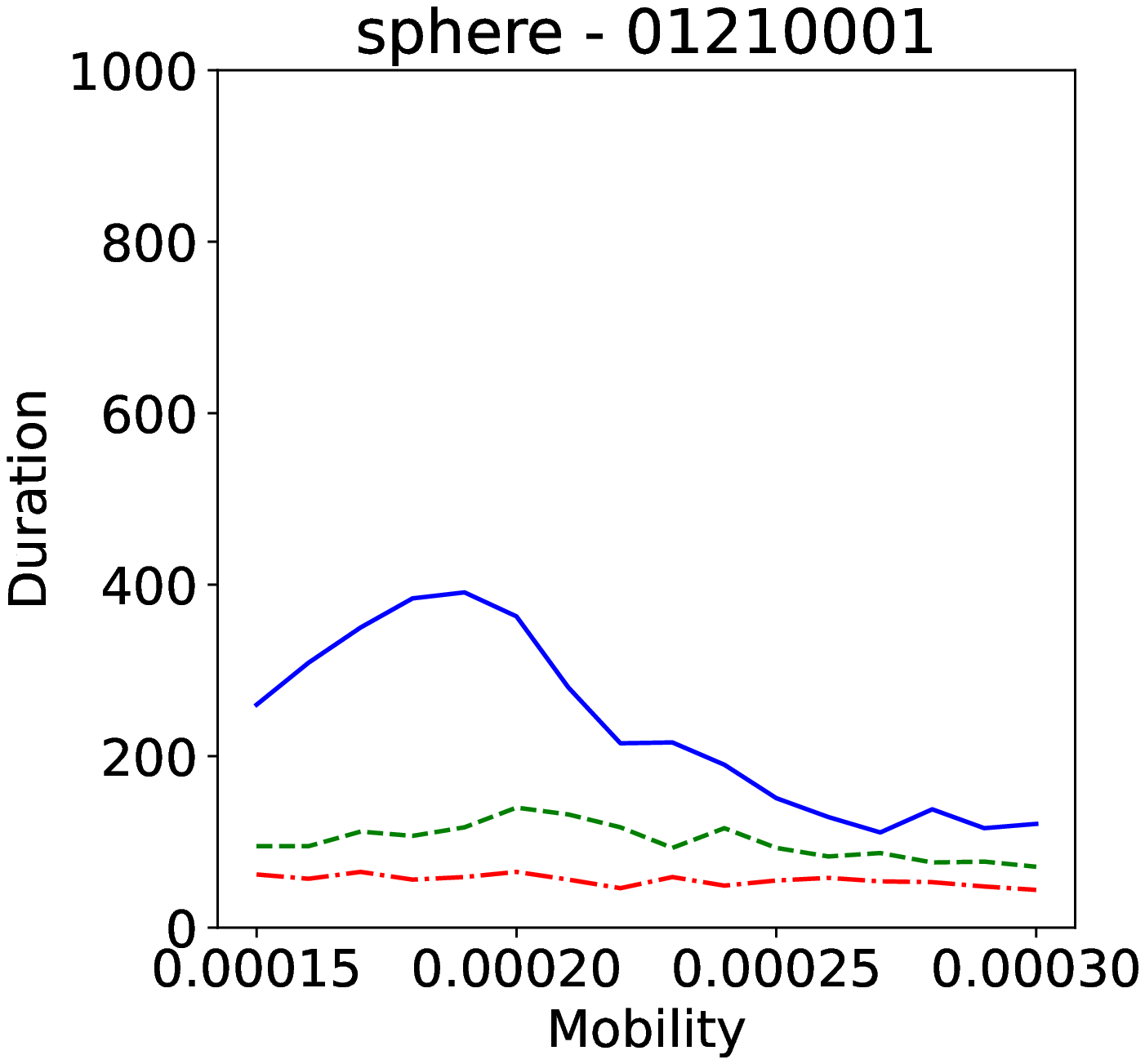}
	\\[\smallskipamount]
	\includegraphics[width=4.7cm]{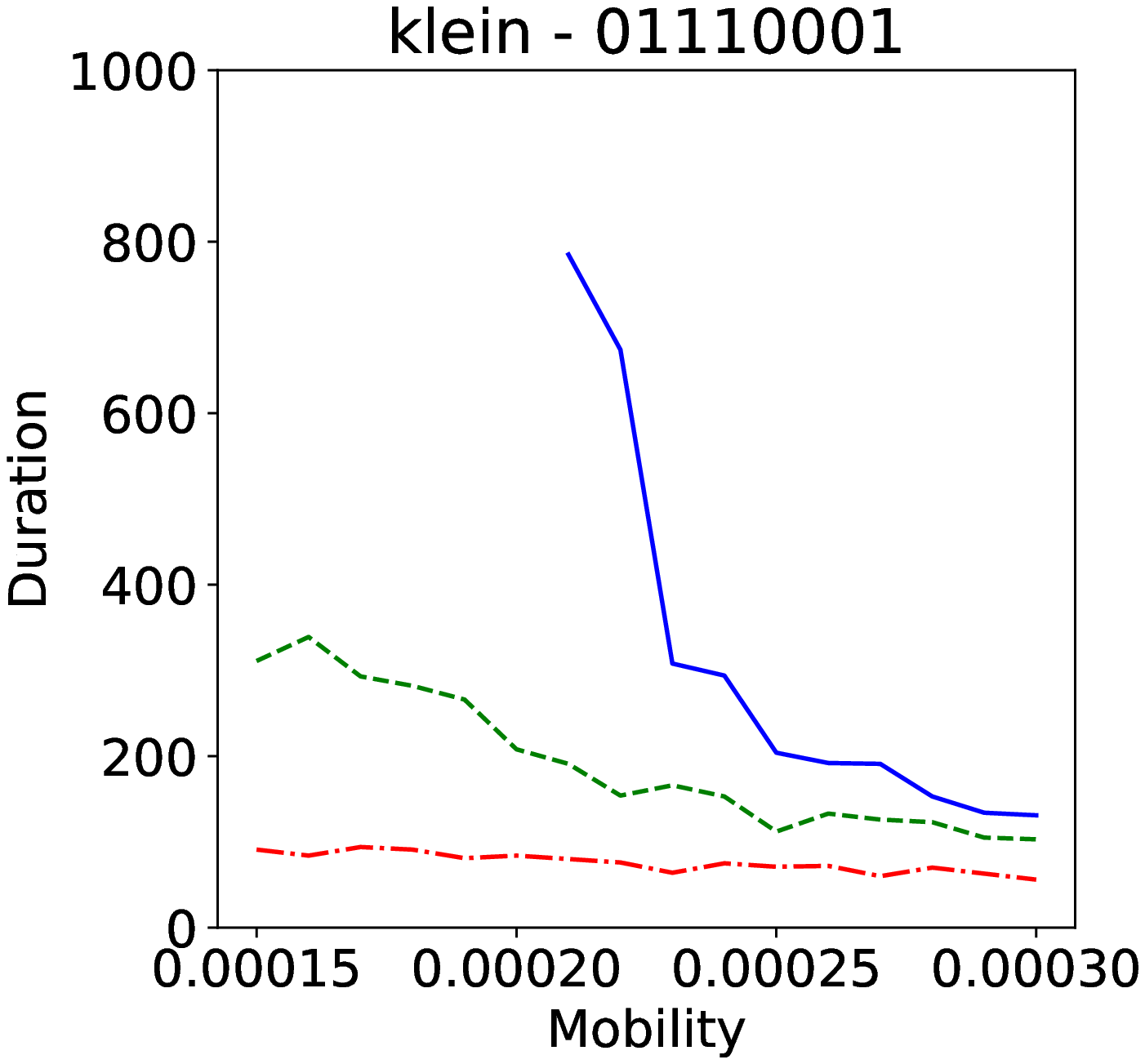}\hfill
	\includegraphics[width=4.7cm]{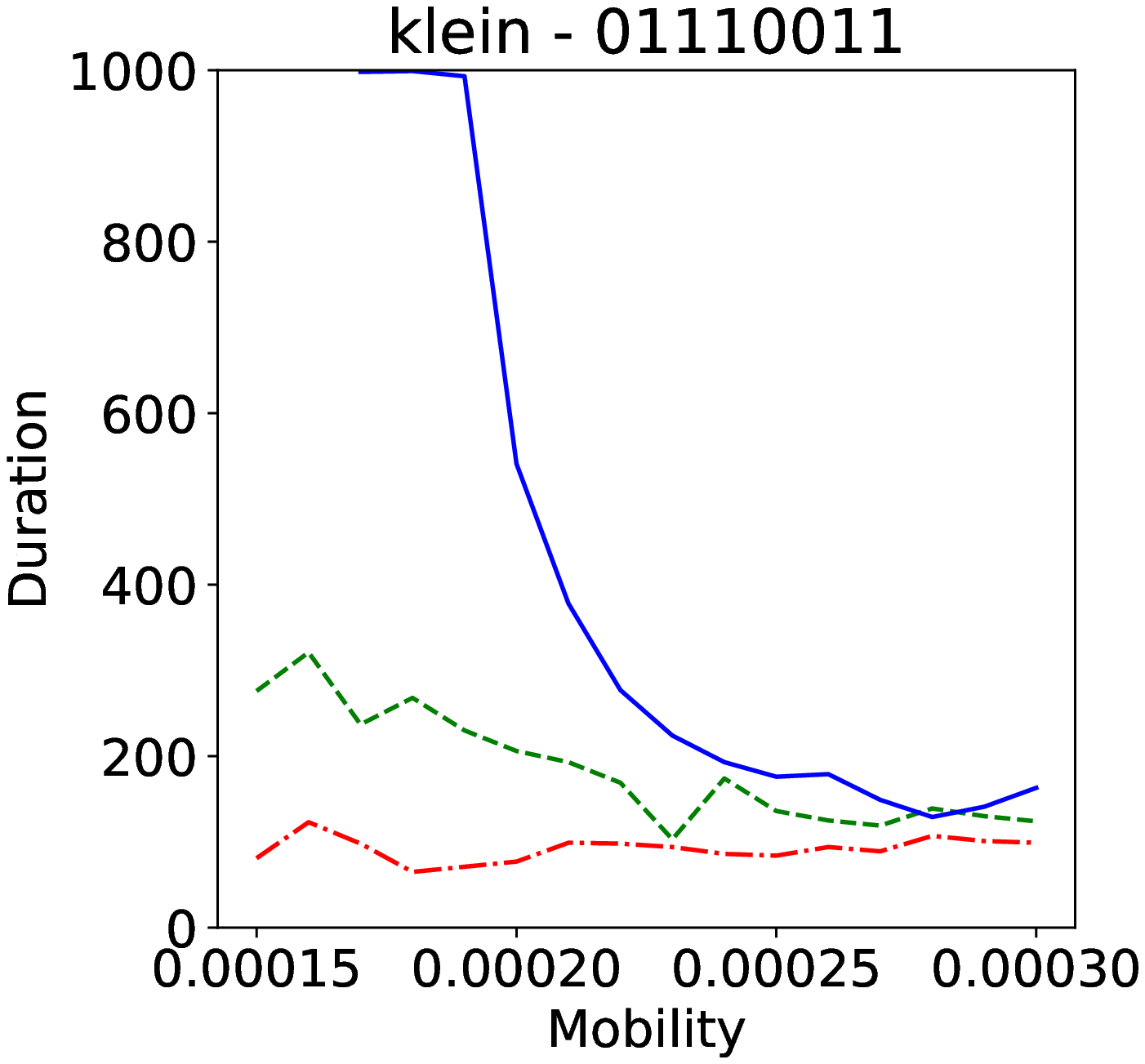}\hfill
	\includegraphics[width=4.7cm]{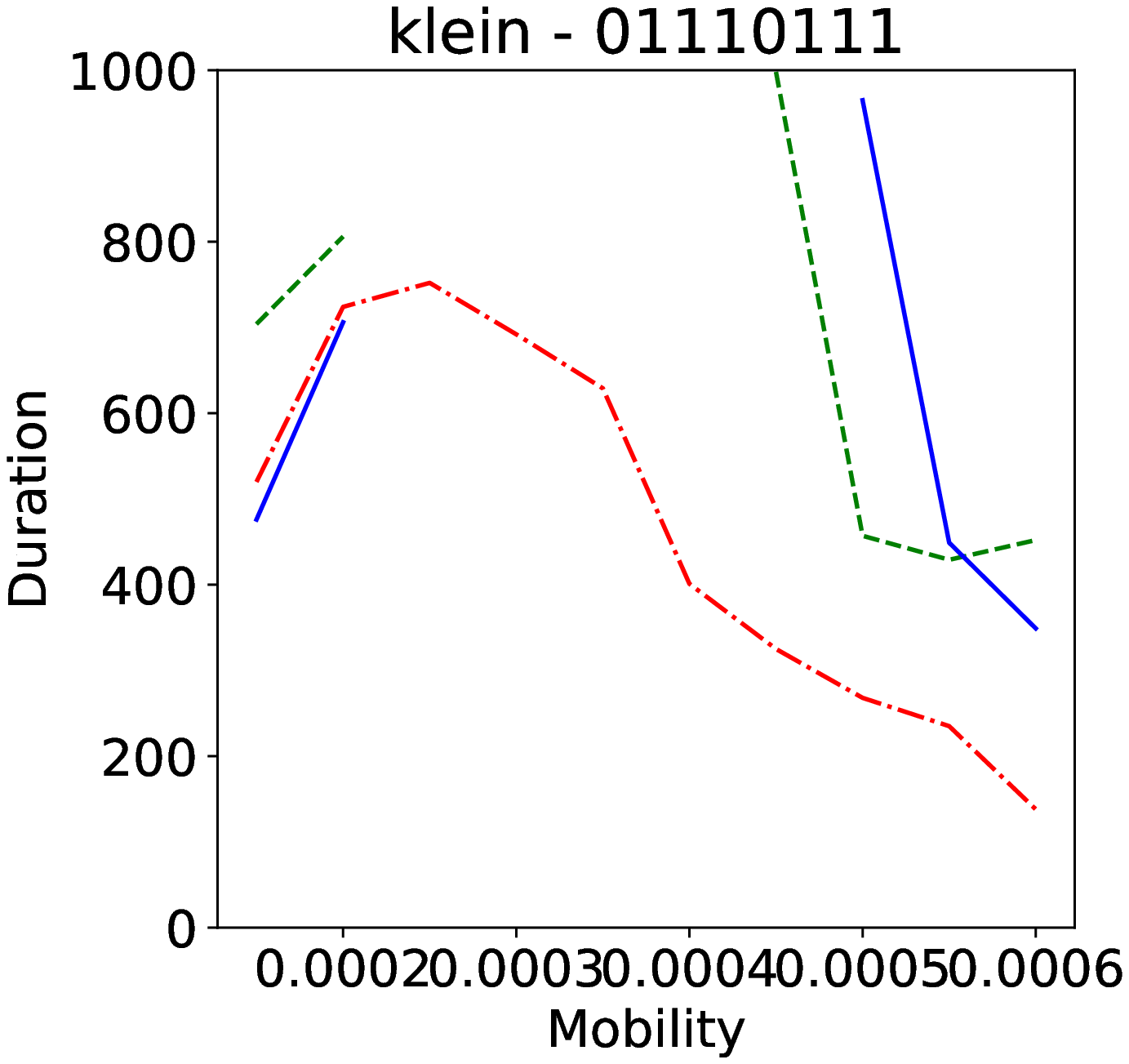}
	\end{center}
	\caption{Median survival times vs mobility for each pattern from experiments 
with patterned initializations. Note range of horizontal axis is not 
consistent. Titles show graph type and pattern code. We omit portions of the 
curve where the median pattern cannot be estimated because fewer than 50\% of 
the experiments had collapsed by 1000 time.
		\label{fig:median-durations}
	}
\end{figure}

\begin{figure}
	\begin{center}
	\includegraphics[width=4.7cm]{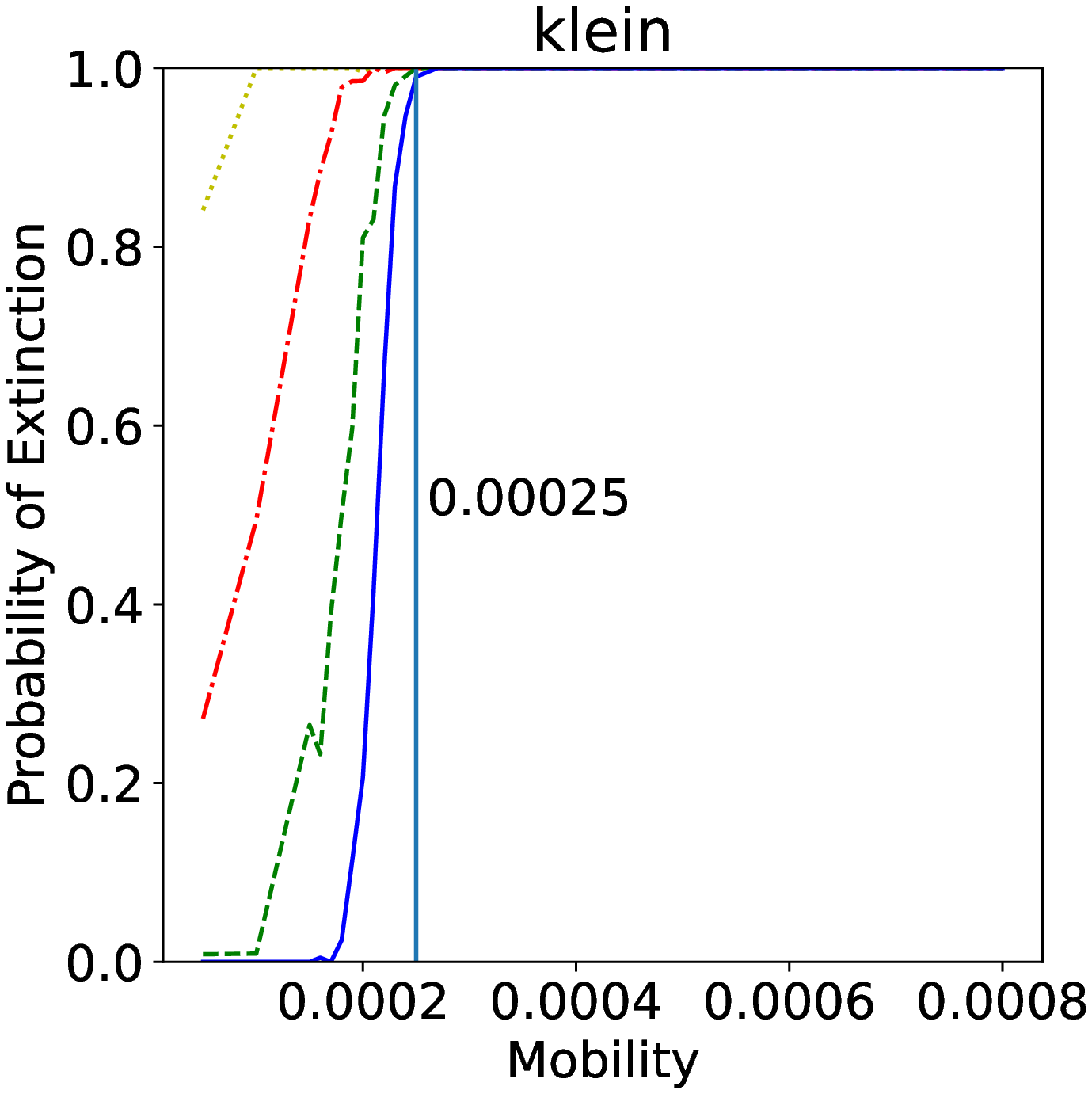}
	\includegraphics[width=4.7cm]{images/extinction_sphere.eps}
	\includegraphics[width=4.7cm]{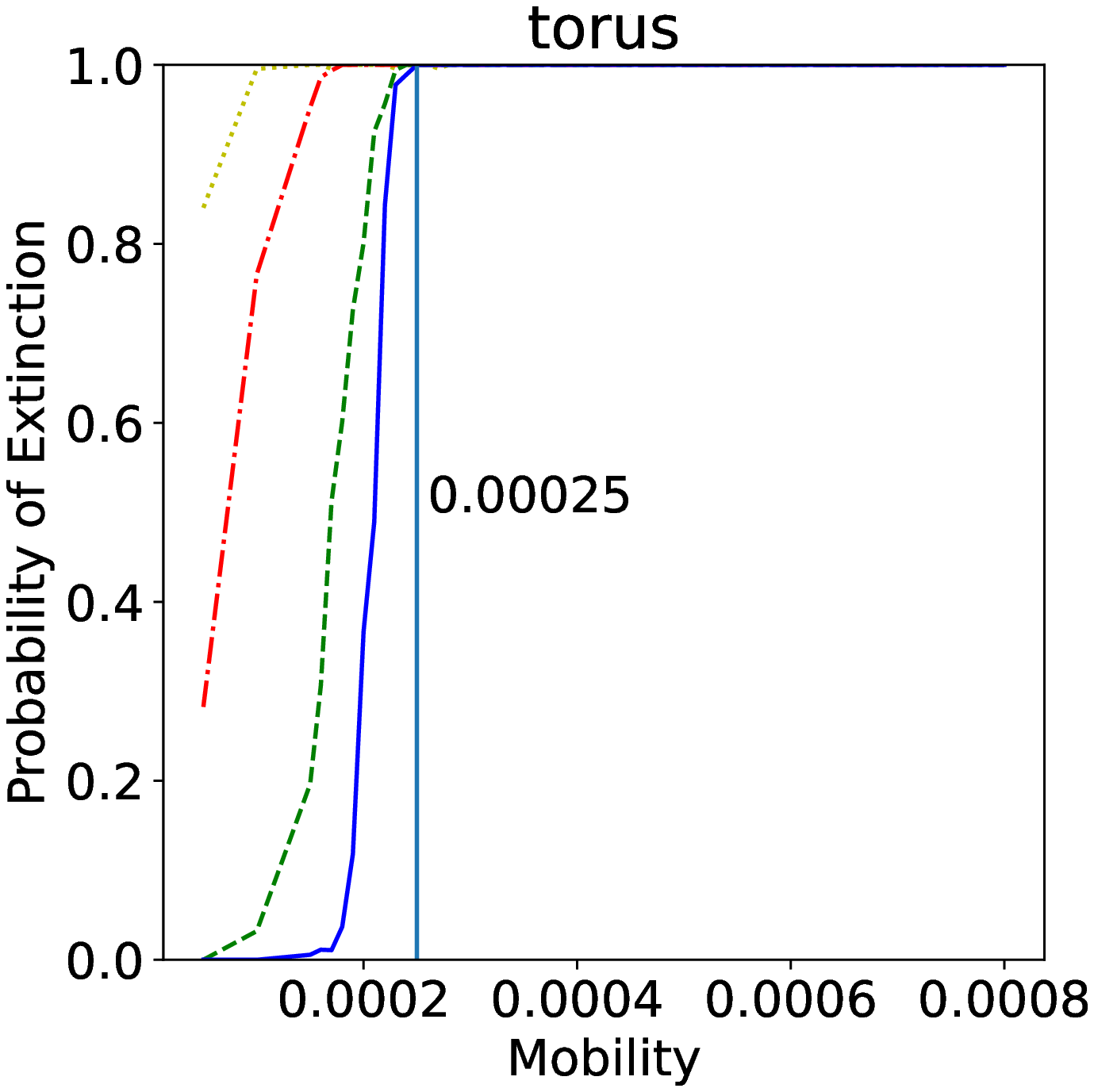}
	\end{center}
	\caption{Probability of extinction by 1000 time versus mobility, from 
experiments with block initialization, excluding experiments that entered the 
marching bands pattern. Each curve represents a different lattice size. 
Vertical lines are drawn at the mobility threshold where $99\%$ of all 
experiments see at least one species go extinct by 1000 time.
		\label{fig:extinction-without-marching-bands}
	}
\end{figure}

\end{document}